\documentclass[11pt,letterpaper]{article}
\usepackage[utf8]{inputenc}
\usepackage{times}
\usepackage{setspace} 
\usepackage{mathtools,wasysym} 
\usepackage{geometry}
\usepackage[inline]{enumitem}
\usepackage{graphicx, wrapfig}
\usepackage[hidelinks]{hyperref}
\usepackage{amsmath, amssymb, amsthm}
\usepackage{cite}
\usepackage{cleveref}
\usepackage[inline]{enumitem}  
\usepackage{xcolor} 
\usepackage[vlined,ruled,linesnumbered]{algorithm2e}
    \SetEndCharOfAlgoLine{}
    \SetKwInput{KwInit}{Init}
\crefalias{algocf}{algorithm}
\usepackage{ifthen} 
\usepackage{mdframed} 
\usepackage{stmaryrd} 

\newcommand{\opt}{\textsc{opt}}
\newcommand{\by}[1]{\text{by #1}}
\newcommand{\floor}[1]{\left\lfloor#1\right\rfloor}
\newcommand{\ceil}[1]{\left\lceil#1\right\rceil}

\newcommand{\polylog}{\mathrm{polylog}}

\newcommand{\group}[1]{\langle #1 \rangle}
\newcommand{\jobgroup}[1]{\llbracket #1 \rrbracket}
\newcommand{\delay}[1]{\rho_{#1}}
\newcommand{\delaymax}{\delay{\max}}
\newcommand{\groupdelay}[1]{\bar{\rho}_{#1}}
\newcommand{\groupspeed}[1]{\bar{s}_{#1}}
\newcommand{\groupsize}[1]{\bar{m}_{#1}}
\newcommand{\delayin}[1]{\rho^{\mathrm{in}}_{#1}}
\newcommand{\delayout}[1]{\rho^{\mathrm{out}}_{#1}}
\newcommand{\ddelay}[1]{d_{#1}}
\newcommand{\ddelayin}[1]{d^{\text{in}}_{#1}}
\newcommand{\ddelayout}[1]{d^{\text{out}}_{#1}}
\newcommand{\dup}{\mathrm{dup}}
\newcommand{\LP}[1]{%
    \ifthenelse{\equal{#1}{a}}{LP\textsubscript{$\alpha$}}{LP\textsubscript{$#1$}}}
\newcommand{\LPimp}[1]{%
    \ifthenelse{\equal{#1}{a}}{$\text{LP}^*_{\alpha}$}{$\text{LP}^*_{#1}$}}
\newcommand{\Calpha}{C^{*}_{\alpha}}
\newcommand{\Csigma}{C_{\sigma}}
\newcommand{\C}[1]{C^*_{#1}}
\newcommand{\makespan}[1]{\Delta_{#1}}
\newcommand{\I}{\mathcal{I}}

\newcommand{\Vlong}{V_{\mathrm{long}}}

\newcommand{\smd}{\textsc{smd}}
\newcommand{\smdjd}{\textsc{smdjd}}
\newcommand{\mdps}{\textsc{mdps}}
\newcommand{\udps}{\textsc{udps}}
\newcommand{\umps}{\textsc{umps}}

\newtheorem{theorem}{Theorem}
\newtheorem{lemma}{Lemma}[section]
\newtheorem{corollary}{Corollary}[theorem]
\newtheorem{claim}{Claim}[lemma]
\newtheorem{thmclaim}{Claim}[theorem]

\newtheorem{definition}[lemma]{Definition}

\theoremstyle{definition}
\newtheorem{alg}[algocf]{Algorithm}

\definecolor{light-gray}{gray}{0.925}
\mdfdefinestyle{theoremstyle}{%
linecolor=light-gray,%
linewidth=2pt,%
innertopmargin=0,
backgroundcolor=light-gray
}
\newmdtheoremenv[
linewidth=0pt,
backgroundcolor=light-gray,
linecolor=light-gray,
innertopmargin=0
]
{result}{Theorem}

\crefalias{result}{theorem}

\newcommand{\linearprogram}[4]{
\noindent
\begin{minipage}[t]{#1 \linewidth}
    \begin{flalign}
        #3
    \end{flalign}
\end{minipage}\hfill
\begin{minipage}[t]{#2 \linewidth}
    \begin{flalign}
        #4
    \end{flalign}
\end{minipage}
\vspace{4mm}
}

\let\llncssubparagraph\subparagraph
\let\subparagraph\paragraph
\usepackage[compact]{titlesec}
\let\subparagraph\llncssubparagraph

\newcommand{\comment}[1]{\textcolor{blue}{#1}}
\newcommand{\junk}[1]{}

\newcommand{\BfPara}[1]{\noindent \textbf{#1.}}

\title{Scheduling under Non-Uniform  Job and Machine  Delays}

\author{
Rajmohan Rajaraman\thanks{Northeastern University, Boston, MA, USA. Email: \texttt{r.rajaraman@northeastern.edu}} 
\hspace{10mm}
David Stalfa \thanks{Northeastern University, Boston, MA, USA. Email: \texttt{stalfa.d@northeastern.edu}} 
\hspace{10mm} 
Sheng Yang \thanks{Independent Researcher. Email: \texttt{styang@fastmail.com}}
}

\date{}

\begin{document}

\maketitle

\begin{abstract}
\noindent
We study the problem of scheduling precedence-constrained jobs on heterogenous machines in the presence of non-uniform job and machine communication delays. We are given as input a set of $n$ unit size precedence-ordered jobs, and a set of $m$ related machines each with size $m_i$ (machine $i$ can execute at most $m_i$ jobs at any time). Each machine $i$ has an associated in-delay $\rho^{\mathrm{in}}_i$ and out-delay $\rho^{\mathrm{out}}_i$. Each job $v$ also has an associated in-delay  $\rho^{\mathrm{in}}_v$ and out-delay $\rho^{\mathrm{out}}_v$. 
In a schedule, job $v$ may be executed on machine $i$ at time $t$ if each predecessor $u$ of $v$ is completed on $i$ before time $t$ or on any machine $j$ before time $t - (\rho^{\mathrm{in}}_i + \rho^{\mathrm{out}}_j + \rho^{\mathrm{out}}_u + \rho^{\mathrm{in}}_v)$. The objective is to construct a schedule that minimizes makespan, which is the maximum completion time over all jobs. 

We consider schedules which allow duplication of jobs as well as schedules which do not.  When duplication is allowed, we provide an asymptotic $\mathrm{polylog}(n)$-approximation algorithm. This approximation is further improved in the setting with uniform machine speeds and sizes. Our best approximation for non-uniform delays is provided for the setting with uniform speeds, uniform sizes, and no job delays. For schedules with no duplication, we obtain an asymptotic $\mathrm{polylog}(n)$-approximation for the above model, and a true $\mathrm{polylog}(n)$-approximation for symmetric machine and job delays. These results represent the first polylogarithmic approximation algorithms for scheduling with non-uniform communication delays. 

Finally, we consider a more general model, where the delay can be an arbitrary function of the job and the machine executing it: job $v$ can be executed on machine $i$ at time $t$ if all of $v$'s predecessors are executed on $i$ by time $t-1$ or on any machine by time $t - \rho_{v,i}$.  We present an approximation-preserving reduction from the Unique Machines Precedence-constrained Scheduling (\textsc{umps}) problem, first defined in [DKRSTZ22], to this job-machine delay model. The reduction entails logarithmic hardness for this delay setting, as well as polynomial hardness if the conjectured hardness of \textsc{umps} holds.

This set of results is among the first steps toward cataloging the rich landscape of problems in non-uniform delay scheduling. 
\end{abstract}

\thispagestyle{empty}
\newpage






    





    
    

    


\pagenumbering{arabic}

\section{Introduction}
With the increasing scale and complexity of scientific and data-intensive computations, it is often necessary to process workloads with many dependent jobs on a network of heterogeneous computing devices with varying computing capabilities and communication delays.  
For instance, the training and evaluation of neural network models, which are iterated directed acyclic graphs (DAGs), is often distributed over diverse devices such as CPUs, GPUs, or other specialized hardware. This process, commonly referred to as \emph{device placement}, has gained significant interest~\cite{mirhoseini2017device,mirhoseini:hierPlace,gao2018spotlight,hafeez2021towards}.  Similarly, many scientific workflows are best modeled as DAGs, and the underlying high-performance computing system as a heterogeneous networked distributed system with communication delays~\cite{andrio2019bioexcel,versluis2018analysis,rico-gallego+dml:hpc}. 

Optimization problems associated with scheduling under communication  delays  have  been  studied  extensively, but  provably  good approximation  bounds  are  few  and  several  challenging  open  problems remain~\cite{ahmad+k:schedule,bampis+gk:schedule,darbha+a:schedule,hoogeveen+lv:schedule,LR02,munier+h:schedule,munier1999approximation,palis1996task,papadimitriou1990towards,picouleau1991two,rayward1987uet}.  With a communication delay, scheduling a DAG of uniform size jobs on identical machines is already
NP-hard~\cite{rayward1987uet,picouleau1991two}, and several inapproximability
results are known~\cite{bampis+gk:schedule,hoogeveen+lv:schedule}.
However, the field is still underexplored and scheduling under communication delay was listed as one of the top ten open problems in scheduling surveys~\cite{bansal:survey,schuurman+w:survey}.  While there has been progress on polylogarthmic-approximation algorithms for the case of uniform communication delays~\cite{LR02,davies2020scheduling,maiti_etal.commdelaydup_FOCS.20,Liu_etal.linear.22}, little is known for more general delay models.

This paper considers the problem of scheduling precedence-constrained jobs on machines connected by a network with \emph{non-uniform} communication delays.  In general, the delay incurred in communication between two machines could vary with the machines as well as with the data being communicated, which in turn may depend on the jobs being excuted on the machines.  For many applications, however, simpler models suffice.  For instance, the machine delays model, where the communication between two machines incurs a delay given by the sum of latencies associated with the two machines, is suitable when the bottleneck is primarily at the machine interfaces.  On the other hand, job delays model scenarios where the delay incurred in the communication between two jobs running on two different machines is a function primarily of the two jobs.  This is suitable when the communication is data-intensive.  Recent work in~\cite{davies_etal.nonuniform_SODA.22} presents a hardness result for a model in which any edge of the DAG separating two jobs running on different machines causes a delay, providing preliminary evidence that obtaining sub-polynomial approximation factors for this model may be intractable.  Given polylogarithmic approximations for uniform delays, a natural question is which, if any, non-uniform delay models are tractable.  

\subsection{Overview of our results}
A central contribution of this paper is to explore and catalog a rich landscape of problems in non-uniform delay scheduling.  We present polylogarithmic approximation algorithms for several models with non-uniform delays, and a hardness result in the mold of~\cite{davies_etal.nonuniform_SODA.22} for a different non-uniform delay model.  Figure~\ref{fig:problem-tree}
organizes various models in this space, with pointers to results in this paper and relevant previous work.
\smallskip

\begin{wrapfigure}{r}{.4\textwidth}
    \centering \vspace{-.1in}
    \includegraphics[width=.35\textwidth, page=1]{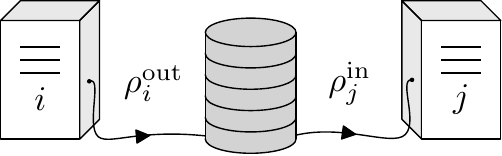}
    \caption{{\small Communicating a result from $i$ to $j$ takes $\delayout i + \delayin  j$ time.}}
\label{fig:machine-comm}
\vspace{-0.1in}
\end{wrapfigure}
\BfPara{Machine delays and job delays (Section~\ref{sec:job_machine_delays})} We begin with a natural model where the delay incurred in communication from one machine to another is the sum of delays at the two endpoints.  Under machine delays, each machine $i$ has an in-delay $\delayin i$ and out-delay $\delayout i$, and the time taken to communicate a result from $i$ to $j$ is $\delayout i + \delayin j$.  This model, illustrated in Figure~\ref{fig:machine-comm}, is especially suitable for environments where data exchange between jobs occurs via the cloud, an increasingly common mode of operation in modern distributed systems~\cite{liang+k:cloud,wu+lthzlzj:cloud,mahgoub+ysmebc:cloud}; $\delayin i$ and $\delayout i$ represent the cloud download and upload latencies, respectively, for machine $i$.

\junk{
\begin{definition}{{\bf (Scheduling under Machine Delays)}}
We are given as input a set of $n$ precedence ordered jobs
and a set of $m$ machines.  
For any jobs $u$ and $v$ with $u \prec v$, machine $i$, and time $t$, we say that $u$ is \textit{available} to $v$ on machine $i$ at time $t$ if $u$ is completed on $i$ before time $t$ or on any machine $j$ before time $t - (\delayout j + \delayin i)$. If job $v$ is scheduled at time $t$ on machine $i$, then all of its predecessors must be available to $v$ on $i$ at time $t$.
The objective of the problem is to construct a schedule that minimizes makespan. \comment{do we still want this definition here if we aren't emphasizing the separate theorem?}
\end{definition}
}

\begin{figure}
    \centering
    \includegraphics[width=\textwidth]{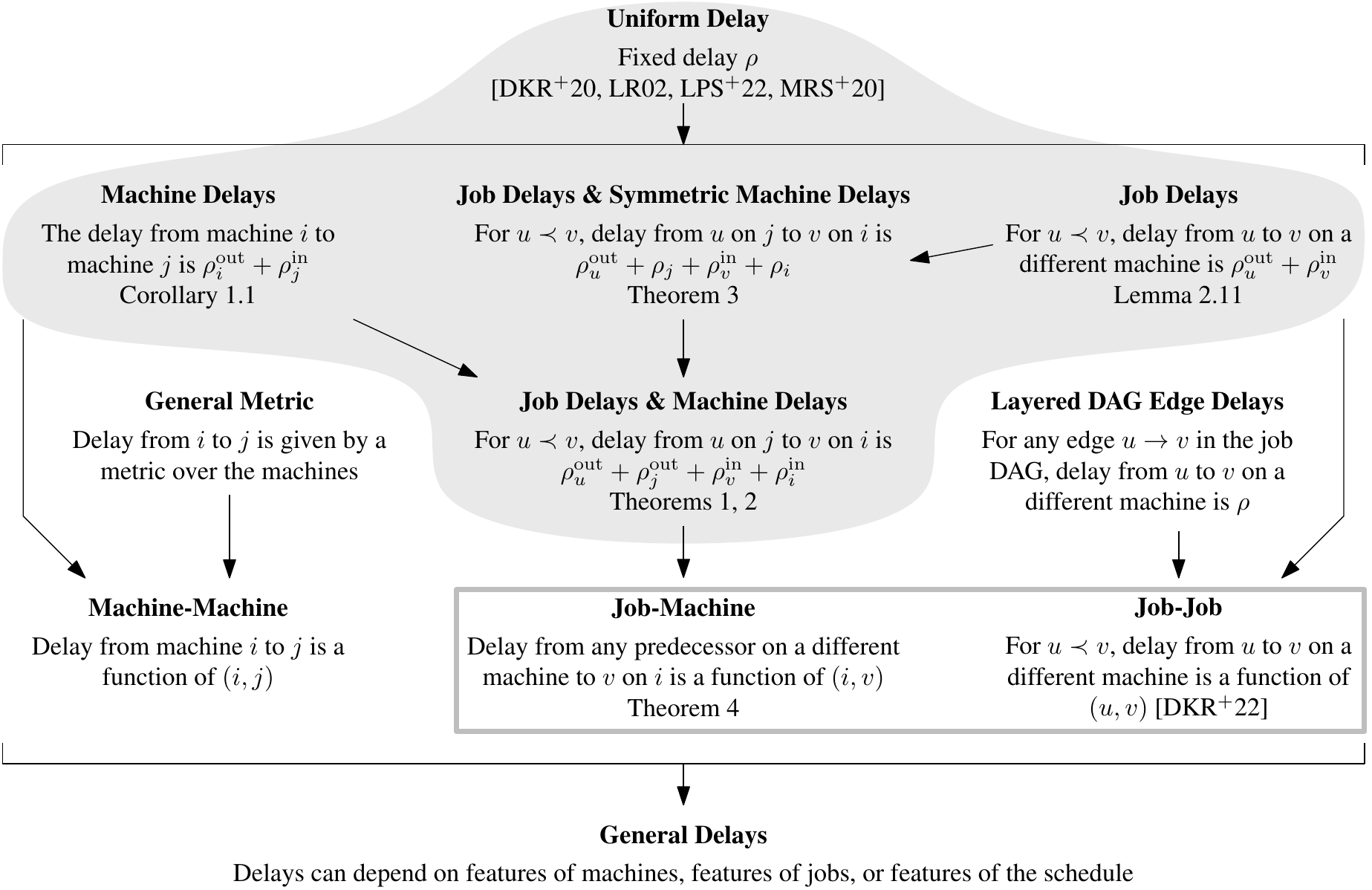}
    \caption{{\small Selection of scheduling models with communication delays. $a \relbar\joinrel\mathrel\RHD b$ indicates that $a$ is a special case of $b$. We present approximation algorithms for models with machine delays and job delays, and a hardness of approximation result for the job-machine delays model.  Theorems and citations point to results in this paper and in previous work, respectively. Those problems backed in gray are ones for which approximation algorithms are known. Those in the gray box are ones for which hardness results have been proven.}}
    \label{fig:problem-tree}
\end{figure}

\junk{
Figure~\ref{fig:problem-tree} situates the machine delay model and other variants we study in this paper within the landscape of communication delay models.
}

\begin{figure}
\includegraphics[width=\textwidth,page=2]{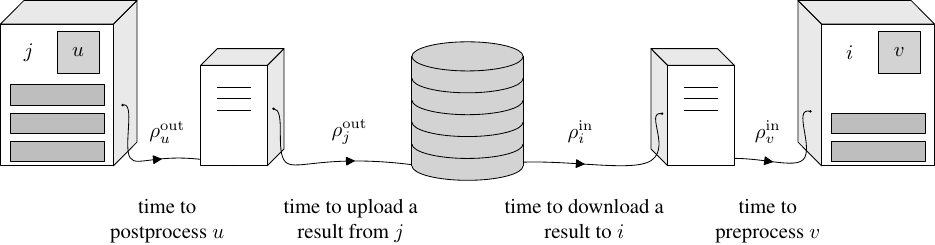}
\caption{{\small Communicating the result of job $u$ on machine $j$ to execute job $v$ on machine $i$.}}
\label{fig:jobmachine-comm}
\end{figure}

\junk{ A natural model is that of a metric.  Special cases of interest are hierarchical metrics or star metrics.  We consider an asymmetric generalization of star metrics, which we refer to as machine-dependent delays.   Motivate this class of problems.
Point to Figure~\ref{fig:problem_tree}.  List theorems. Give high level approach, emphasizing simplicity. Maybe introduce via Chudak-Shmoys simplification? We can also talk about how this approach is also simpler than FOCS paper. Using those more complicated techniques, we can achieve better approximations in settings with uniform size clusters.}

\junk{
\begin{definition}
For any schedule $\sigma$, let $\makespan{\sigma}$ denote the makespan of $\sigma$, and let $\delay{\sigma}$ denote $\max\{ \delayin  i + \delayout i\}$ over all machines $i$ on which a job is executed in $\sigma$.
\end{definition}
}

The machine delays model does not account for heterogeneity among jobs, where different jobs may be producing or consuming different amounts of data, which may impact the delay between the processing of one job and that of another dependent job on a different machine.  To model this, we allow each job $u$ to have an in-delay $\delayin u$ and an out-delay $\delayout u$.  
\begin{definition}{{\bf (Scheduling under Machine Delays and Job Delays)}}
We are given as input a set of $n$ precedence ordered jobs
and a set of $m$ machines.  For any jobs $u$ and $v$ with $u \prec v$, machine $i$, and time $t$, $u$ is \textit{available} to $v$ on machine $i$ at time $t$ if $u$ is completed on $i$ before time $t$ or on any machine $j$ before time $t - (\delayout j + \delayout u + \delayin i + \delayin v)$.  (This model is illustrated in Figure~\ref{fig:jobmachine-comm}.)  If job $v$ is scheduled at time $t$ on machine $i$, then all of its predecessors must be available to $v$ on $i$ at time $t$.  We define $\delay{\max} = \max_{x \in V \cup M}\{\delayin x + \delayout x\}$.  The objective is to construct a schedule that minimizes makespan.
\label{def:smdjd}
\end{definition}

We present the first approximation algorithms for scheduling under non-uniform communication delays.  In the presence of delays, a natural approach to hide latency and reduce makespan is to duplicate some jobs (for instance, a job that is a predecessor of many other jobs)~\cite{papadimitriou1990towards,ahmad+k:schedule}.  We consider both schedules that allow duplication (which we assume by default) and those that do not.  Our first result is a polylogarithmic asymptotic approximation for scheduling under machine and job delays when duplication is allowed.  
\begin{result}
There exists a polynomial time algorithm for scheduling under machine and job delays, that produces a schedule with makespan $O((\log^9 n)(\opt + \delay{\max}))$.%
\label{thm:smdjd}
\end{result}
We emphasize that if the makespan of any schedule includes the delays incurred in distributing the problem instance and collecting the output of the jobs, then the algorithm of Theorem~\ref{thm:smdjd} is, in fact, a \emph{true polylogarithmic approximation} for makespan.  (From a practical standpoint, in order to account for the time incurred to distribute the jobs and collect the results, it is natural to include in the makespan the in- and out-delays of every machine used in the schedule.)

\smallskip

\BfPara{Related machines and multiprocessors (Section~\ref{sec:general})} 
Theorem~\ref{thm:smdjd} is based on a new linear programming framework for addressing non-uniform job and machine delays.  We demonstrate the power and flexibility of this approach by 
incorporating two more aspects of heterogeneity: speed and number of processors.  Each machine $i$ has a number $m_i$ of processors and a speed $s_i$ at which each processor processes jobs.  We generalize Theorem~\ref{thm:smdjd} to obtain the following main result of the paper.
\begin{result}
There exists a polynomial time algorithm for scheduling on related multiprocessor machines under machine and job delays, that yields a schedule with makespan $\polylog(n) (\opt + \delay{\max}))$.%
\label{thm:smdjd-general}
\end{result}
  The exact approximation factor obtained depends on the non-uniformity of the particular model.  For the most general model we consider in Theorem~\ref{thm:smdjd-general}, our proof achieves a $O(\log^{15} n)$ bound.  We obtain improved bounds when any of the three defining parameters---size, speed, and delay---are uniform.  
For instance, we obtain an approximation factor of $O(\log^5 n)$ for scheduling uniform speed and uniform size machines under machine delays alone, i.e., when there are no job delays (Corollary~\ref{cor:machine} of Section~\ref{sec:job_machine_delays}).  Further, with only job delays and uniform machine delays, we provide a combinatorial asymptotic $O(\log^6 n)$ approximation (Lemma~\ref{thm:job} of Section~\ref{sec:job_machine_delays}) which is improved to an asymptotic $O(\log n)$ approximation if the input contains no out-delays.  We note that despite some uniformity, special cases can model certain two-level non-uniform network hierarchies with processors at the leaves, low delays at the first level, and high delays at the second level. 

\junk{
The next theorem presents improved bounds for special cases.
\begin{result}[\textbf{Improved asymptotic approximation for \mdps\ when delays or machine sizes are uniform}]
Suppose that either the number of processors or the delays are uniform across all machines. There exists a polynomial time algorithm that produces a schedule with makespan at most $O((\log n)^{2} \cdot (\makespan{\sigma} + \delay{\sigma}))$ for any schedule $\sigma$. 
\end{result}
The analysis is more complicated, uses optimization from \cite{chudak1999approximation}.
}

\smallskip
\BfPara{No-duplication schedules (Section~\ref{sec:no-dup})}
We next consider the problem of designing schedules that do not allow duplication.  We obtain a polylogarithmic asymptotic approximation via a reduction to scheduling with duplication.  Furthermore, if the delays are symmetric (i.e., $\delayout{i} = \delayin{i}$ for all $i$, and $\delayout v = \delayin v$ for all $v$) we are able to find a \emph{true} polylogarithmic-approximate no-duplication schedule.  To achieve this result, we present an approximation algorithm to estimate if the makespan of an optimal no-duplication schedule is at least the delay of any given machine; this enables us to identify machines that cannot communicate in the desired schedule.\footnote{We note that the corresponding problem for duplication schedules is a min-max partitioning variant of the Minimum $k$-Union problem and related to the Min-Max Hypergraph $k$-Partitioning problem, both of which have been shown to be Densest-$k$-Subgraph-hard~\cite{chlamtavc2017minimizing,chandrasekaran+c:partition}; this might suggest a similar hardness result for deriving a \emph{true} approximation for \mdps\ with duplication.}
\begin{result}
\label{thm:no-dup}
There exists a polynomial time algorithm for scheduling on related multiprocessor machines under machine delays and job delays, which produces a no-duplication schedule with makespan $\polylog(n) (\opt + \delay{\max})$.  If $\delayin  i = \delayout i$ for all $i$, then there exists a polynomial time $\polylog(n)$-approximation algorithm for no-duplication schedules.
\end{result}

\junk{
Note that the first three theorems imply approximation algorithms for the model where there is a central processor that commmunicates other machines. This model requires an initial communication \textit{to} all machines that execute jobs, and a final communication phase \textit{from} all machines that execute jobs. In this model, the makespan of any schedule will be at least the $\delayin  i + \delayout i$ if any job is executed on machine $i$.

[Talk about the challenges involved in removing the additive factor]
\begin{result}[\textbf{Polylogarithmic approximation for \smdjd\ with symmetric delays and no duplication}]
\label{thm:mdps-symmetric-no-dup}
If $\delayin  i = \delayout i$ for all $i$, then there exists a polynomial time $\polylog(n)$-approximation algorithm for \smdjd\ with no duplication.
\end{result}
}
\junk{For cases where the optimal makespan is at least the max communication, the previous theorem gives an asymptotic approximation. So, in order to prove this theorem, we present a way to check if opt is less than a given value. Note that, for this setting, the previous theorems give asymptotic approximations -- they minimize $\max_{i,t} \{t + \delayout i + \delayin  i:$ some job is executed on $i$ at time $t\}$.}
\smallskip

\BfPara{Pairwise delays (Section~\ref{sec:umps_reduction})}
All of the preceding results concern models where the communication associated with a precedence relation $u \prec v$ when $u$ and $v$ are executed on different machines $i$ and $j$ is an \emph{additive} combination of delays at $u$, $v$, $i$, and $j$.  Additive delays are suitable for capturing independent latencies incurred by various components of the system.  A more general class of models considers \emph{pairwise} delays where the delay is an \emph{arbitrary function} of $i$ and $j$ (machine-machine), $u$ and $v$ (job-job), or either job and the machine on which it executes (job-machine). 
The machine-machine delay model captures classic networking scenarios, where the delay across machines is determined by the network links connecting them.  Job-job delays model applications where the data that needs to be communicated from one job to another descendant job depends arbitrarily on the two jobs.  The job-machine model is well-suited for applications where the delay incurred for communicating the data consumed or produced by a job executing on a machine is an arbitrary function of the size of the data and the bandwidth of the machine.  Recent work in~\cite{davies_etal.nonuniform_SODA.22} shows that scheduling under job-job delays is as hard as the Unique Machine Precedence Scheduling ({\textsc UMPS}) problem, providing preliminary evidence that obtaining sub-polynomial approximation factors may be intractable.  We show that {\textsc UMPS} also reduces to scheduling under job-machine delays, suggesting a similar inapproximability for this model. 
\begin{result}[\textbf{\umps\ reduces to scheduling under job-machine delays}]
\label{thm:mdps-symmetric-no-dup}
There is a polynomial-time approximation-preserving reduction from \umps\ to the scheduling under job-machine delays.
\end{result}


\subsection{Overview of our techniques}
\label{sec:techniques}
Our approximation algorithms for scheduling under job delays and machine delays (Theorem~\ref{thm:smdjd} proved in Section~\ref{sec:job_machine_delays}) and the generalization to related machines and multiprocessors (Theorem~\ref{thm:smdjd-general} proved in Section~\ref{sec:general}) rely on a framework composed of a carefully crafted linear programming relaxation and a series of reductions that help successively reduce the level of heterogeneity in the problem. While each individual component of the framework refines established techniques or builds on prior work, taken together they offer a flexible recipe for designing approximation algorithms for scheduling precedence-ordered jobs on a distributed system of heterogeneous machines with non-uniform delays. Given the hardness conjectures of \cite{davies_etal.nonuniform_SODA.22} for the job-job delay setting (and for the job-machine setting via Theorem~\ref{thm:umps}), we find it surprising that a fairly general model incorporating both job delays and machine delays on related machines is tractable. 

Previous results on scheduling under (uniform) communication delays are based on three different approaches: (a) a purely combinatorial algorithm of~\cite{LR02} that works only for uniform delay machines; (b) an LP-based approach of~\cite{maiti_etal.commdelaydup_FOCS.20} that handles related machines and uniform delays, assuming jobs can be duplicated, and then extends to no-duplication via a reduction; and (c) an approach of~\cite{davies2020scheduling} based on a Sherali-Adams hierarchy relaxation followed by a semi-metric clustering, which directly tackles the no-duplication model. 
At a very high level, our main challenge, which is not addressed in any of the previous studies, is to tackle the \emph{multi-dimensional heterogeneity} of the problem space: in the nature of delays (non-uniform values, in- and out-delays, job delays, machine delays) as well as the machines (delay, speed, and size).  

We pursue an LP-based framework, which significantly refines the approach of~\cite{maiti_etal.commdelaydup_FOCS.20}.  Their algorithm  organizes the computation in phases, each phase corresponding to a (uniform) delay period, and develops a linear program that includes delay constraints capturing when jobs have to be phase-separated and phase constraints bounding the amount of computation within a phase.  In non-uniform delay models, the delay constraints for a job $v$ executing on a machine $i$ depend not only on the predecessors of $v$, but also on the machines on which they may be scheduled. While there is a natural way to account for non-uniform in-delays in the LP, incorporating out-delays or even symmetric delays poses technical difficulties.  We overcome this hurdle by first showing that out-delays can be eliminated by suitably adjusting in-delays, at the expense of a polylogarithmic factor in approximation, thus allowing us to focus on in-delays.  

Despite the reduction to in-delays, extending the LP of~\cite{maiti_etal.commdelaydup_FOCS.20} by replacing the uniform delay parameter by the non-uniform delay parameters of our models fails and yields a high integrality gap.  This is because their algorithm crucially relies on an ordering of the machines (on the basis of their speeds), which is exploited both in the LP (in the delay and phase constraints) as well as how jobs get assigned and moved in the computation of the final schedule.  Given the multi-dimensional heterogeneity of the problems we study, there is no such natural ordering of the machines.  To address the above hurdle, we organize the machines and jobs into groups based on their common characteristics (delay, speed, size), and introduce new variables for assigning jobs to groups without regard to any ordering among them.  This necessitates new load and delay constraints and a change in rounding and schedule construction.  We now elaborate on these ideas, as we discuss our new framework in more detail.  
\smallskip

\BfPara{Reduction to in-delays (Section~\ref{sec:out-in})}
The first ingredient of our recipe is an argument that any instance of the problem with machine delays and job delays can be reduced to an instance in which all out-delays are 0, meaning that in the new instance delays depend only on the machine and job receiving the data, at the expense of a polylogarithmic factor in approximation.  This reduction is given in Lemma~\ref{lem:reduction} and Algorithm~\ref{alg:reduction} in Section~\ref{sec:out-in}.  To convert from a given schedule with out-delays to one without, we subtract $\delayout i + \delayout v$ from the execution time of every job $v$ on machine $i$. However, in order to avoid collisions, we expand the given schedule into phases of different length, organized in particular sequence so that the execution times within each phase may be reduced without colliding with prior phases. This transforms the schedule into one where the in-delay of every machine $i$ is $\delayin i + \delayout i$ and every job $v$ is $\delayin v + \delayout v$.  This transformation comes at a constant factor cost for machine delays and an $O(\log^2 \delay{max})$ cost for job delays.  A similar procedure converts from an in-delay schedule to one with in- and out-delays, completing the desired reduction.
\smallskip

\BfPara{The linear program (Sections~\ref{sec:partition_jm}-\ref{sec:LP_jm} and~\ref{sec:partition}-\ref{sec:LP})}
Before setting up the linear program, we partition the machines and the jobs into groups of uniform machines and jobs, respectively; i.e. each machine in a group can be treated as having the same in-delay, speed, and size (to within a constant factor), and each job in a group can be treated as having the same in-delay.  The final approximation factor for the most general model grows as $K^3$ and $L$, where $L$ is the number of job groups and $K$ is the number of machine groups, which depends on the extent of heterogeneity among the machines.  We bound $K$ by $O(\log^3 n)$ in the case when the speeds, sizes, and delays of machines are non-uniform.  We emphasize that, even with the machines partitioned in this way, we must carefully design our LP to judiciously distribute jobs among the groups depending on the precedence structure of the jobs and the particular job and machine parameters. 

\junk{Our LPs seek an optimal (fractional) placement of the jobs (with possible duplication) on the machine groups.  They contain job placement, duplication, and completion time variables, and incorporate non-uniform in-delays for both machines and jobs.  In Section~\ref{sec:general}, we also include machine speed and size parameters. \comment{add more here}
One challenge is bounding the amount of duplication allowed within a communication phase of a particular group. To this end, we incorporate constraints that capture the optimal makespan for scheduling the duplicated jobs on uniform machines. This strategy points to the interesting possibility of capturing more complex processor structure, such as might be modelled with a multi-level tree hierarchy.}
Our LP is inspired by that of \cite{maiti_etal.commdelaydup_FOCS.20}, though significant changes are necessary to allow for non-uniform delays. The key constraints of each LP are presented below (with the constraints from \cite{maiti_etal.commdelaydup_FOCS.20} rewritten to include machine group variables). 
\begin{center}
\includegraphics[width=\textwidth]{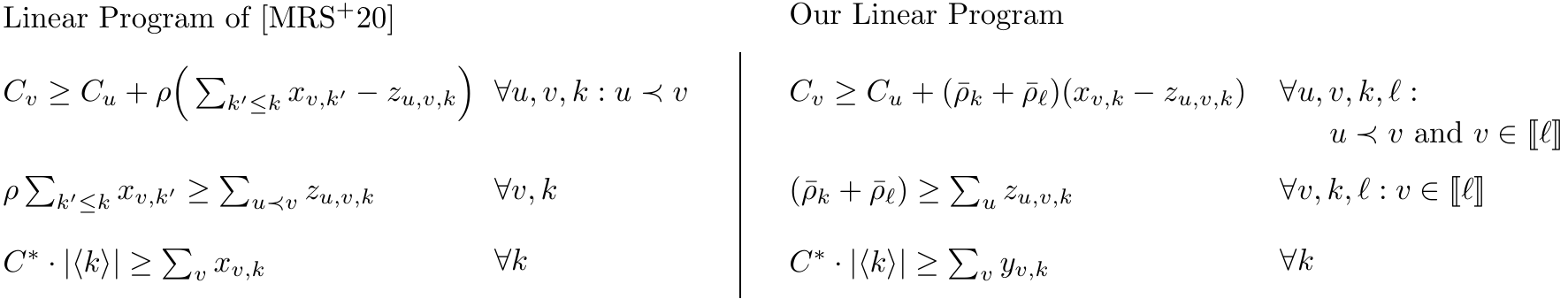}
\end{center}
Here, $C^*$ represents the makespan of the schedule and $C_v$ represents the earliest execution time of job $v$. $x_{v,k}$ indicates if $v$ is placed on a machine in group $\group k$ ($=1$) or not ($ = 0$). $z_{u,v,k}$ indicates whether $x_{v,k} = 1$ and $C_v - C_u$ is less the time it takes to communicate the result of $u$ from a different machine. $y_{v,k}$  takes the maximum of $x_{v,k}$ and $\max_{u}\{z_{v,u,k}\}$ to indicate whether some copy of $v$ is executed on a machine in group $\group k$ ($=1$) or not ($=0$). Other notation used in the linear program is explained in Section~\ref{sec:job_machine_delays}.

We can see that the LP of \cite{maiti_etal.commdelaydup_FOCS.20} relies heavily on an ordering of the machine groups by speed. Given the multidimensional heterogeneity of machines in our most general setting, such an ordering is not available to us. We address this problem by eliminating the use of a machine order, and assign jobs to machine groups directly. Because we use machine group variables, and each problem has $\polylog(n)$ machine groups, the delays and total load are not increased by more than a factor of $\polylog(n)$. However, a problem arises in construction of the final schedule.

In \cite{maiti_etal.commdelaydup_FOCS.20}, the ordering of the groups was leveraged to construct the final schedule by always placing a job on higher capacity groups than the one to which it is assigned by the LP. Since the LP assigns all jobs to some group, we can infer that the total load over all groups does not increase by more than a constant factor. With multidimensional heterogenous machines, there is no clear ordering of machine groups to achieve a similar property (e.g.\ one set of jobs may be highly parallelizable, while another requires a single fast machine). Our solution is to place all jobs on those groups to which the LP assigns them, along with any predecessors indicated by the $z$-variables. However, such a construction could vastly exceed the value of the LP unless the load contributed by the $z$-variables is counted toward the LP makespan. To this end, we introduce the $y$-variables and associated constraints, which account for this additional, duplicated load.
In the most general setting, we also introduce constraints which govern the amount of duplication possible within a single communication phase. These additional constrains model an optimal schedule of the duplicated jobs on the uniform machines within a single group.

\smallskip
\BfPara{Rounding the LP solution and determining final schedule (Sections~\ref{sec:round_jm}-\ref{sec:schedule_jm} and~\ref{sec:round}-\ref{sec:schedule})}
The next component rounds an optimal LP solution to an integer solution by placing each job on the group for which the job's LP mass is maximized. We also place duplicate predecessors of each job $v$ on its group according to the $z$-variables for $v$'s predecessors. This indicates a key difference with \cite{maiti_etal.commdelaydup_FOCS.20}, where the load contributed by duplicates was handled by the ordering of the machines. A benefit of our simple rounding is that it accommodates many different machine and job properties as long as the number of groups can be kept small. Finally, we construct a schedule using the integer LP solution. This subroutine divides the set of jobs assigned to each group into phases and constructs a schedule for each phase by invoking a schedule for the uniform machines case, appending each schedule to the existing schedule for the entire instance. 
\smallskip

\BfPara{No-duplication schedules (Section~\ref{sec:no-dup})}
The proof of the first part of Theorem~\ref{thm:no-dup} extends an asympototic polylogarithmic approximation to no-duplication schedules for machine delays and job delays. The theorem follows from the structure of the schedule designed in Theorem~\ref{thm:smdjd-general} and a general reduction in~\cite{maiti_etal.commdelaydup_FOCS.20} from duplication to no-duplication schedules in the uniform delay case.  Avoiding the additive delay penalty of the first part of Theorem~\ref{thm:no-dup} to achieve a true approximation is much more difficult.  When delays are symmetric (i.e., in-delays equal out-delays), we can distinguish those machines whose delay is low enough to communicate with other machines from those machines with high delay. One of the central challenges is then to distribute jobs among the high-delay machines. We overcome this difficulty by revising the LP in the framework of Theorem~\ref{thm:smdjd-general} to partition the jobs among low- and high-delay machines, and rounding the corresponding solutions separately.

We then must distinguish between those jobs with delay low enough to communicate with other jobs from those with high delay. We note that any predecessor or successor of a high delay job must be executed on the same machine as that job. We leverage this fact to construct our schedule, first placing all high delay jobs with their predecessors and successors on individual machines. We then run our machine and job delay algorithm with the remaining jobs on the low delay machines. This schedule is placed after the execution of the downward closed high-delay components, and before the upward closed high-delay components, ensuring that the schedule is valid.  

We note that  the design of no-duplication schedules via a reduction to duplication schedules incurs a loss in approximation factor of an additional polylogarithmic factor.  While this may not be desirable in a practical implementation, our results demonstrate the flexibility of the approach and highlight its potential for more general delay models.
\smallskip

\BfPara{Hardness for job-machine delay model (Section~\ref{sec:umps_reduction})}
The algorithmic framework outlined above incorporates non-uniform job and machine delays that combine additively.  It is natural to ask if the techniques extend to other delay combinations or more broadly to pairwise delay models.  In the job-machine delay model we study, when a job $u$ executed on machine $i$ precedes job $v$ executed on machine $j$, then a delay $\rho_{v,j}$ between the two executions is incurred.  Our reduction from \umps\ to the job-machine delay problem follows the approach of~\cite{davies_etal.nonuniform_SODA.22} by introducing new jobs with suitable job-machine delay parameters that essentially force each job to be executed on a particular machine.  This reduction does not require the flexibility of assigning different delays for different job-job pairs, but it is unclear if the same technique can be applied to machine-machine delay models.  Delineating the boundary between tractable models and those for which polylogarithmic approximations violate conjectured complexity lower bounds is a major problem of interest.

\subsection{Related work}

\BfPara{Precedence constrained scheduling}
The problem of scheduling precedence-constrained jobs was initiated in the classic work of Graham who gave a constant approximation algorithm for uniform machines~\cite{graham:schedule}.  Jaffe presented an $O(\sqrt m)$ makespan approximation for the case with related machines~\cite{Jaffe_1980}.  This was improved upon by Chudak and Shmoys who gave an $O(\log m)$ approximation~\cite{chudak1999approximation}, then used the work of Hall, Schulz, Shmoys, and Wein~\cite{Hall_1997} and Queyranne and Sviridenko~\cite{Queyranne_2002} to generalize the result to an $O(\log m)$ approximation for weighted completion time. Chekuri and Bender~\cite{chekuri-bender01} proved the same bound as Chudak and Shmoys using a combinatorial algorithm.  In subsequent work, Li improved the approximation factor to $O(\log m/ \log \log m)$~\cite{Li17}.  The problem of scheduling precedence-constrained jobs is hard to approximate even for identical machines, where the constant depends on complexity assumptions~\cite{LenstraK78, Bansal_2009,Svensson_2010}.
Also, Bazzi and Norouzi-Fard \cite{Bazzi_2015} showed a close connection between structural hardness for \(k\)-partite graph and scheduling with precedence constraints.

\smallskip
\BfPara{Precedence constrained scheduling under communication delays}
Scheduling under communication delays has been studied extensively~\cite{rayward1987uet, papadimitriou1990towards, veltman+ll:schedule}. For unit size jobs, identical machines, and unit delay, a ($7/3$)-approximation is given in \cite{munier+h:schedule}, and \cite{hoogeveen+lv:schedule} proves the NP-hardness of achieving better than a $5/4$-approxmation. Other hardness results are given in \cite{bampis+gk:schedule, picouleau1991two, rayward1987uet}. More recently, Davies, Kulkarni, Rothvoss, Tarnawski, and Zhang \cite{davies2020scheduling} give an $O(\log \rho \log m)$ approximation in the identical machine setting using an LP approach based on Sherali-Adams hierarchy, which is extended to include related machines in~\cite{davies_etal.commdelayrelated_SODA.21}. Concurrently, Maiti, Rajaraman, Stalfa, Svitkina, and Vijayaraghavan \cite{maiti_etal.commdelaydup_FOCS.20} provide a polylogarithmic approximation for uniform communication delay with related machines as a reduction from scheduling with duplication. The algorithm of \cite{maiti_etal.commdelaydup_FOCS.20} is combinatorial in the case with identical machines.  

Davies, Kulkarni, Rothvoss, Sandeep, Tarnawski, and Zhang \cite{davies_etal.nonuniform_SODA.22} consider the problem of scheduling precedence-constrained jobs on uniform machine in the presence of non-uniform, job-pairwise communication delays. That is, if $u \prec v$ and $u$ and $v$ are scheduled on different machines, then the time between their executions is at least $\delay{u,v}$. The authors reduce to this problem from Unique-Machines Precedence-constrained Scheduling (\umps) in which there is no communication delay, but for each job there is some particular machine on which that job must be placed. The authors show that \umps\ is hard to approximate to within a logarithmic factor by a reduction from job-shop scheduling, and conjecture that \umps\ is hard to approximate within a polynomial factor.

\smallskip
\BfPara{Precedence constrained scheduling under communication delays with job duplication}
Using duplication with communication delay first studied by Papadimitriou and Yannakakis~\cite{papadimitriou1990towards}, who give a 2-approximation for DAG scheduling with unbounded processors and fixed delay. Improved bounds for infinite machines are given in \cite{ahmad+k:schedule, darbha+a:schedule, munier+k:schedule, palis1996task}. Approximation algorithms are given by Munier and Hanen \cite{munier1999approximation, munier+h:schedule} for special cases in which the fixed delay is very small or very large, or the DAG restricted to a tree. The first bounds for a bounded number of machines are given by Lepere and Rapine~\cite{LR02} who prove an asymptotic $O(\log \rho/\log \log \rho)$ approximation. Recent work has extended their framework to other settings:  \cite{maiti_etal.commdelaydup_FOCS.20} uses duplication to achieve an $O(\log \rho \log m/ \log \log \rho)$ approximation for a bounded number of related machines, and Liu, Purohit, Svitkina, Vee, and Wang~\cite{Liu_etal.linear.22} improve on the runtime of~\cite{LR02}  to a near linear time algorithm with uniform delay and identical machines.


\subsection{Discussion and open problems}

Our results indicate several directions for further work. First, we conjecture that our results extend easily to the setting with non-uniform job sizes. We believe the only barriers to such a result are the techinical difficulties of tracking the completion times of very large jobs that continue executing long after they are placed on a machine. Also, while our approximation ratios are the first polylogarithmic guarantees for scheduling under non-uniform delays, we have not attempted to optimize logarithmic factors.  There are obvious avenues for small reductions in our ratio, e.g. the technique used in \cite{LR02} to reduce the ratio by a factor of $\log \log \rho$.  More substantial reduction, however, may require a novel approach. Additionally, in the setting without duplication, we incur even more logarithmic factors owing to our reduction to scheduling with duplication. These factors may be reduced by using a more direct method, possibly extending the LP-hierarchy style approach taken in \cite{davies2020scheduling, davies_etal.commdelayrelated_SODA.21}.

Aside from improvements to our current results, our techniques suggest possible avenues to solve related non-uniform delay scheduling problems. Our incorporation of parallel processors allows our results to apply to a two-level machine hierarchy, where machines are the leaves and the delay between machines is a function of their lowest common ancestor.  We would like to explore extensions of our framework to constant-depth hierarchies and tree metrics.  More generally, scheduling under metric and general machine-machine delays remains wide open (see Figure~\ref{fig:problem-tree}). 

Finally, we believe there are useful analogs to these machine delay models in the job-pairwise regime. A job $v$ with in-delay $\delayin{v}$ and out-delay $\delayout{v}$ has the natural interpretation of the data required to execute a job, and the data produced by a job. A job tree hierarchy could model the shared libraries required to execute certain jobs: jobs in different subtrees require different resources to execute, and downloading these additional resources incurs a delay. Given the hardness conjectures of~\cite{davies_etal.nonuniform_SODA.22} and our hardness result for the job-machine delay model, further refining Figure~\ref{fig:problem-tree} and exploring the tractability boundary would greatly enhance our understanding of scheduling under non-uniform delays. Additionally, our notion of job delays essentially depends on the precedence relation over the jobs. Another natural notion of job delay may be to consider a DAG defined over the jobs, with a delay incurred only if there is a directed edge $u \to v$ (rather than $u \prec v$).  For general DAGs, the result of \cite{davies_etal.nonuniform_SODA.22} shows this setting to be at least as hard as \umps, which the authors conjecture to be hard to approximate to within a polynomial factor, but good approximations may be achievable for special cases such as layered DAGs.






\section{Machine Delays and Job Delays}
\label{sec:job_machine_delays}

In this section, we present an asymptotic approximation algorithm for scheduling under machine delays and job delays for unit speed and size machines.  As discussed in Section~\ref{sec:techniques}, we can focus on the setting with no out-delays, at the expense of a polylogarithmic factor in approximation; Lemma~\ref{lem:reduction} of Section~\ref{sec:out-in} presents the reduction to in-delays.  Therefore, in this section, we assume that $\delayout i = 0$ for all machines $i$ and $\delayout v = 0$ for all jobs $v$.  For convenience, we use $\delay i$ to denote the in-delay $\delayin i$ of machine $i$ and $\delay v$ to denote the in-delay $\delayin v$ of machine $v$.  Let $\delaymax = \max \{ \max_v\{\delay v\}, \max_i \{ \delay i \} \}$. 

\subsection{Partitioning machines and jobs into groups}
\label{sec:partition_jm}
In order to simplify our exposition and analysis, we introduce a new set of machines $M'$ with rounded delays. For each $i \in M$, if $2^{k-1} \le \delay i < 2^{k}$, we introduce $i' \in M'$ with $\delay{i'} = 2^k$. We then partition $M'$ according to machine delays: machine $i \in M'$ is in $\group k$ if $\delay i = 2^k$; we set $\groupdelay k = 2^k$. We also introduce a new set of jobs $V'$ with rounded delays. For each $v \in V$, if $2^{\ell-1} \le \delay v < 2^{\ell}$, we introduce $v' \in V'$ with $\delay{v'} = 2^{\ell}$.  We then partition $V'$ according to job delays: job $v \in V'$ is in $\jobgroup{\ell}$ if $\delay v = 2^{\ell} = \groupdelay{\ell}$.   For the remainder of the section, we work with the machine set $M'$ and the job set $V'$, ensuring that all machines or jobs within a group have identical delays.  As shown in the following lemma, this partitioning is at the expense of at most a constant factor in approximation.  

\begin{lemma}
    The optimal makespan over the machine set $V', M'$ is no more than a factor of 2 greater than the optimal solution over $V, M$. 
\label{lem:group_rounding_jm}
\end{lemma}
\begin{proof}
    Consider any schedule $\sigma$ on the machine set $M$. We first show that increasing the delay of each machine by a factor of 2 increases the makespan of the schedule by at most a factor of 2. We define the schedule $\sigma'$ as follows. For every $i,t$, if $(i,t) \in \sigma(v)$, then $(i,2t) \in \sigma'(v)$. It is easy to see that $\sigma'$ maintains the precedence ordering of jobs, and that the time between the executions of any two jobs has been doubled. Therefore, $\sigma'$ is a valid schedule with all communication delays doubled, and with the makespan doubled.
\end{proof}
We can assume that $\max_k \{\groupdelay k\} \le n$ since if we ever needed to communicate to a machine with delay greater than $n$ we could schedule everything on a single machine in less time. Therefore, we have $K \le \log n$ machine groups. Similarly, $\max_{\ell} \{ \groupdelay{\ell}\} \le n$, implying that we have $L \le \log n$ job groups.

\subsection{The linear program}
\label{sec:LP_jm}
In this section, we design a linear program \LP{a}---Equations~(\ref{LPjm:makespan}-\ref{LPjm:z>0})---parametrized by  $\alpha \ge 1$, for machine delays.  Following Section~\ref{sec:partition_jm}, we assume that the machines and jobs are organized in groups, where each group $\group k$ (resp., $\jobgroup{\ell}$) is composed of machines (resp., jobs) that have identical delay.
 
\vspace{-1.5em}
\linearprogram{.6}{.3}{
    &\Calpha \ge C_v &&\forall v 
    \label{LPjm:makespan} 
    \\
    &\Calpha \cdot |\group k| \ge \sum_{v} y_{v,k} &&\forall k 
    \label{LPjm:load} 
    \\
    &C_v \ge C_u + (\groupdelay k + \groupdelay{\ell}) (x_{v,k} - z_{u,v,k} ) &&\forall u,v,k,\ell: 
    \label{LPjm:delay} 
    \\
    &&& \quad u \prec v, v \in \jobgroup{\ell}
    \notag
    \\
    &C_v \ge C_u + \sum_k x_{v,k} &&\forall u,v: u \prec v
    \label{LPjm:precedence} 
    \\
    &\alpha (\groupdelay k + \groupdelay{\ell}) \ge \sum_u z_{u,v,k} &&\forall v,k,\ell: v \in \jobgroup{\ell}
    \label{LPjm:duplicates}
}{
    &\sum_{k} x_{v,k} = 1 &&\forall v
    \label{LPjm:execution}
    \\[.35em]
    &C_v \ge 0 && \forall v
    \label{LPjm:C>0} 
    \\[.35em]
    &x_{v,k} \ge z_{u,v,k} &&\forall u,v, k
    \label{LPjm:x>z} 
    \\[.35em]
    &y_{v,k} \ge x_{v,k} &&\forall v,k
    \label{LPjm:y>x} 
    \\[.35em]
    &y_{u,k} \ge z_{u,v,k} &&\forall u,v,k
    \label{LPjm:y>z} 
    \\[.35em]
    &z_{u,v,k} \ge 0 &&\forall u,v,k
    \label{LPjm:z>0} 
}

\BfPara{Variables} 
$\Calpha$ represents the makespan of the schedule. For each job $v$, $C_v$ represents the earliest completion time of $v$. For each job $v$ and group $\group k$, $x_{v,k}$ indicates whether or not $v$ is first executed on a machine in group $\group k$. For each $\group k$ and pair of jobs $u,v$ such that $u \prec v$ and $v \in \jobgroup{\ell}$, $z_{u,v,k}$ indicates whether $v$ is first executed on a machine in group $\group k$ and the earliest execution of $u$ is less that $\groupdelay k + \groupdelay \ell$ time before the execution of $v$.  Intuitively, $z_{u,v,k}$ indicates whether there must be a copy of $u$ executed on the same machine that first executes $v$. For each job $v$ and group $\group k$, $y_{v,k}$ indicates whether $x_{v,k} = 1$ or $z_{u,v,k} = 1$ for some $u$; that is, whether or not some copy of $v$ is placed on group $\group k$. Constraints~(\ref{LPjm:C>0} - \ref{LPjm:z>0}) guarantee that all variables are non-negative.

\smallskip
\BfPara{Makespan (\ref{LPjm:load}, \ref{LPjm:makespan})}
Constraint~\ref{LPjm:makespan} states that the makespan is at least the maximum completion time of any job.
Constraint~\ref{LPjm:load} states that the makespan is at least the load on any single group. 

\smallskip
\BfPara{Delays (\ref{LPjm:delay}, \ref{LPjm:duplicates})}
Constraint~\ref{LPjm:delay} states that the earliest completion time of $v \in \jobgroup{\ell}$ must be at least $\groupdelay k + \groupdelay \ell$ after the earliest completion time of any predecessor $u$ if $v$ is first executed on a machine in group $\group k$ and no copy of $u$ is duplicated on the same machine as $v$.  Constraint~\ref{LPjm:duplicates} limits the amount of duplication that can be done to improve the completion time of any job: if $v \in \jobgroup{\ell}$ first executes on a machine in group $\group k$ at time $t$, then the number of predecessors that may be executed in the $\groupdelay k + \groupdelay \ell$ steps preceding $t$ is at most $\groupdelay k$. 

The remaining constraints enforce standard scheduling conditions. Constraint~\ref{LPjm:precedence} states that the completion time of $v$ is at least the completion time of any of its predecessors, and constraint~\ref{LPjm:execution} ensures that every job is executed on some group.  Constraints~\ref{LPjm:execution} and \ref{LPjm:x>z} guarantee that $z_{u,v,k} \le 1$ for all $u,v,k$. This is an important feature of the LP, since a large $z$-value could be used to disproportionately reduce the delay between two jobs in constraint~\ref{LPjm:delay}.  

\begin{lemma}{\textbf{(\LP 1 is a valid relaxation)}}
    The minimum of $\C 1$ is at most $\opt$.
\label{lem:LPjm_valid}
\end{lemma}
\begin{proof}
    Consider an arbitrary schedule $\sigma$ with makespan $\Csigma$, i.e.\ $\Csigma = \max_{v,i,t}\{ t: (i,t) \in \sigma(v)\}$. 
    
    \smallskip
    \BfPara{LP solution} Set $\C 1 = \Csigma$. For each job $v$, set $C_v$ to be the earliest completion time of $v$ in $\sigma$, i.e.\ $C_v = \min_{i,t}\{t : (i,t) \in \sigma(v)\}$. Set $x_{v,k} = 1$ if $\group k$ is the group that contains the machine on which $v$ first completes (choosing arbitrarily if there is more than one) and 0 otherwise. For $u,v,k$, set $z_{u,v,k} = 1$ if $u \prec v$, $x_{v,k} = 1$, $v \in \jobgroup{\ell}$, and $C_v - C_u < \groupdelay k + \groupdelay{\ell}$ (0 otherwise). Set $y_{u,k} = \max\{x_{u,k}, \max_v \{ z_{u,v,k} \}\}$.
    
    \smallskip
    \BfPara{Feasibility} We now establish that the solution defined is feasible. Constraints (\ref{LPjm:makespan}, \ref{LPjm:C>0}--\ref{LPjm:z>0}) are easy to verify. We now establish constraints~(\ref{LPjm:load}--\ref{LPjm:duplicates}). Consider constraint~\ref{LPjm:load} for fixed group $\group k$. $\sum_v y_{v,k}$ is upper bound by the total load $\Lambda$ on $\group k$. The constraint follows from $\Calpha \ge \Csigma \ge \Lambda / |\group k|$.
    
    Consider constraint~\ref{LPjm:delay} for fixed $u,v,k$ where $u \prec v$. Let $X = x_{v,k}$ and let $Z = z_{u,v,k}$. If $(X,Z) = (0,0), (0,1)$, or $(1,1)$ then the constraint follows from constraint~\ref{LPjm:precedence}. If $(X,Z) = (1,0)$, then by the assignment of $z_{u,v,k}$ we can infer that $C_v - C_u \ge \groupdelay k + \groupdelay{\ell}$, which shows the constraint is satisfied.
    
    Consider constraint~\ref{LPjm:duplicates} for fixed $v,k$. If $x_{v,k} = 0$ then the result follows from the fact that $z_{u,v,k} = 0$ for all $u$. If $x_{v,k} =  1$, then we can infer that $v \in \jobgroup{\ell}$. So, at most $\groupdelay k + \groupdelay{\ell}$ predecessors of $v$ that can be scheduled in the $\groupdelay k + \groupdelay{\ell}$ time before $C_v$, ensuring that the constraint is satisfied.  
\end{proof}

\subsection{Deriving a rounded solution to the linear program}
\label{sec:round_jm}

\begin{definition}
    $(C,x,y,z)$ is a \emph{rounded solution} to \LP a if all values of $x,y,z$ are either 0 or 1. 
\end{definition}

Let \LP 1 be defined over machine groups $\group 1, \group 2, \ldots, \group K$ and job groups $\jobgroup 1, \jobgroup 2, \ldots, \jobgroup L$. Given a solution $(\hat C, \hat x, \hat y, \hat z)$ to \LP 1, we construct an integer solution $(C,x,y,z)$ to \LP{2K} as follows. For each $v,k$, set $x_{v,k} = 1$ if $k = \max_{k'}\{\hat x_{v,k'}\}$ (if there is more than one maximizing $k$, arbitrarily select one); set to 0 otherwise. Set $z_{u,v,k} = 1$ if $x_{v,k} = 1$ and $\hat z_{u,v,k} \ge 1/(2K)$; set to 0 otherwise. For all $u,k$, $y_{u,k} = \max\{ x_{u,k}, \max_v \{z_{u,v,k}\}\}$. Set $C_v = 2K \cdot \hat C_v$. Set $\C{2K} = 2K \cdot \hat{\C 1}$. 

\begin{lemma}
    If $(\hat C, \hat x, \hat y, \hat z)$ is a valid solution to \LP 1, then  $(C,x,y,z)$ is a valid solution to \LP{2K}. 
\label{lem:LPKjm_solution}
\end{lemma}
\begin{proof}
    By constraint~(\ref{LPjm:execution}), $\sum_{k} \hat x_{v,k}$ is at least 1, so $\max_{k} \{\hat x_{v,k}\}$ is at least $1/K$.  Therefore, $x_{v,k} \le K \hat{x}_{v,k}$ for all $v$ and $k$. Also, $z_{u,v,k} \le 2K \hat z_{u,v,k}$ for any $u,v,k$ by definition. By the setting of $C_v$ for all $v$, $y_{v,k}$ for all $v,k$, and $\C{2K}$, it follows that constraints~(\ref{LPjm:makespan}, \ref{LPjm:precedence}-\ref{LPjm:z>0}) of \LP{1} imply the respective constraints of \LP{2K}.
    We first establish constraint~(\ref{LPjm:load}). For any fixed group $\group k$,
    \begin{align*}
        2K \hat{C}_1 \cdot |\group k| &\ge 2K \sum_v \hat y_{v,k} = 2K \sum_v \max\{ \hat x_{v,k}, \max_u \{ \hat z_{v,u,k}\} \}  &&\by{constraints~\ref{LPjm:load}, \ref{LPjm:z>0} of \LP 1} \\
        &\ge 2K \sum_v \frac{ x_{v,k} + \max_u \{ z_{v,u,k} \}}{2K} \ge \sum_v y_{v,k} &&\by{definition of $y_{v,k}$} 
    \end{align*}
    which entails constraint~(\ref{LPjm:load}) by $\C{2K} = 2K \hat{\C{1}}$.
    It remains to establish constraint~(\ref{LPjm:delay}) for fixed $u,v,k$.  We consider two cases. If $x_{v,k} - z_{u,v,k} \le 0$, then the constraint is trivially satisfied 
    in \LP{2K}. If $x_{v,k} - z_{u,v,k} = 1$, then, by definition of $x$ and $z$, 
    $\hat x_{v,k} - \hat z_{u,v,k}$ is at least $1/(2K)$. This entails that $\hat C_v \ge \hat C_u + ((\groupdelay k + \groupdelay{\ell})/{2K})$ which establishes constraint~(\ref{LPjm:delay}) of \LP{2K} by definition of $C_v$ and $C_u$.
\end{proof}

\begin{lemma}
    $C_{2K} \le 4  K \cdot \opt$.
\label{lem:LPKjm_makespan}
\end{lemma}
\begin{proof}
    Lemma~\ref{lem:group_rounding_jm} shows that our grouping of machines does not increase the value of the LP by more than a factor of 2. Therefore, by Lemmas~\ref{lem:LPjm_valid} and \ref{lem:LPKjm_solution}, $C_{2K} = 2K \cdot \hat C_1 \le 4 K \cdot \opt$. 
\end{proof}

\subsection{Computing a schedule given an integer solution to the LP}
\label{sec:schedule_jm}

Suppose we are given a partition of $M$ into $K$ groups such that group $\group k$ is composed of identical machines (i.e. for all $i,j \in \group k$, $\delay i = \delay j$). Also, suppose we are given a partition of $V$ into $L$ groups such that group $\jobgroup{\ell}$ is composed of jobs with identical in-delay. Finally, we are given a rounded solution $(C,x,y,z)$ to \LP a defined over machine groups $\group 1, \ldots, \group K$ and job groups $\jobgroup 1, \ldots \jobgroup L$. In this section, we show that we can construct a schedule that achieves an approximation for machine delays in terms of $\alpha, K$, and $L$. The combinatorial subroutine that constructs the schedule is defined in Algorithm~\ref{alg:scheduler_jm}.  In the algorithm, we use a subroutine \udps-Solver for Uniform Delay Precedence-Constrained Scheduling.  An $O(\log \delay{}/\log\log \delay{})$-asympototic approximation is given in~\cite{LR02}.  For completeness, we use the \udps-Solver presented and analyzed in Section~\ref{sec:uniform}, which generalizes the algorithm of~\cite{LR02} to incorporate non-uniform machine sizes.

\junk{
\paragraph{Uniform Delay Precedence-constrained Scheduling (\udps).}
We are given a set of unit size partially ordered jobs, a set of identical machines, and a delay $\delta$. In a schedule, if job $v$ is executed on machine $i$ at time $t$, then the result of $v$ is available to $i$ at time $t+1$ and to all other machines at time $t + \delta$. In order to execute $v$ on $i$ at time $t$, all of $v$'s predecessors must be available to $i$ at time $t$. The goal is to construct a schedule to minimize makespan. Duplication is allowed.

Note that \udps\ is 

Note that \udps\ is identical to machine delays and job precedence delays if, for all $i,j$, $\delayin i = \delayout i = \delayin j  = \delayout j$ and, for all $u,v$, $\delayin u = \delayout u = \delayin v = \delayout v$. This is equivalent to DAG scheduling with a single, uniform delay. , we give an algorithm \udps-Solver for \udps, which is obtained by generalizing .  This algorithm is taken from \cite{LR02} and slightly generalized to fit our framework.
}

\begin{algorithm}
\setstretch{1.25}
\KwInit{$\forall v, \ \sigma(v) \gets \varnothing;\ T \gets 0;\ \theta \gets 0$}
\While{$T \le \Calpha$}{
    \ForAll{machine groups $\group k$}{
        \For{job group $\jobgroup{\ell} = \jobgroup L$ to $\jobgroup 1$: $\exists$ integer $d,\ T = d (\groupdelay k + \groupdelay{\ell})$}{
                $V_{k,\ell,d} \gets \{v \in \jobgroup{\ell}: x_{v,k} = 1$ and $T \le C_v < T+ \groupdelay k + \groupdelay{\ell} \} $ \label{line:Vkld}\;
                $U_{k,\ell,d} \gets \{u : \exists v \in V_{k,d},\  u \prec v$ and $T \le C_u < T + \groupdelay k + \groupdelay{\ell} \}$ \label{line:Ukld}\;
                $\sigma' \gets $ \udps-Solver on $(V_{k,\ell,d} \cup U_{k,\ell,d}, \group k, \groupdelay k + \groupdelay{\ell} )$\;
                $\forall v,i,t,\ $ if $(i,t) \in \sigma'(v)$ then $\sigma(v) \gets \sigma(v) \cup \{ (i, \theta + \groupdelay k + \groupdelay{\ell} + t)\}$ \label{line:append_jm} \label{line:append}\;
                $\theta \gets \theta + 2(\groupdelay k + \groupdelay{\ell})$
        }
    }
    $T \gets T + 1$\;
}
\caption{Machine Delay Scheduling with Duplication}
\label{alg:scheduler_jm}
\label{alg:scheduler}
\end{algorithm}

We now describe Algorithm~\ref{alg:scheduler_jm} informally. The subroutine takes as input the rounded \LP a solution $(C,x,y,z)$ and initializes an empty schedule $\sigma$ and global parameters $T,\theta$ to 0. For a fixed value of $T$, we iterate through all machine groups $\group k$ and job groups $\jobgroup{\ell}$, with decreasing $\ell$. For a fixed value of $T, k, \ell$, we check if there is some integer $d$ such that $T = d (\groupdelay k + \groupdelay{\ell})$. If so, we define $V_{k,\ell,d}$ and $U_{k,\ell,d}$ as in lines \ref{line:Vkld} and \ref{line:Ukld}. 
We then call \udps-Solver to construct a \udps\ schedule $\sigma'$ on jobs $V_{k,\ell,d} \cup U_{k,\ell,d}$, machines in $\group k$, and delay $\groupdelay k + \groupdelay{\ell}$.  We then append $\sigma'$ to $\sigma$. Once all values of $k,\ell$ have been checked, we increment $T$ and repeat until all jobs are scheduled. The structure of the schedule produced by Algorithm~\ref{alg:scheduler_jm} is depicted in Figure~\ref{fig:schedule_structure}. Lemma~\ref{lem:uniform_jm} (entailed by Lemma~\ref{lem:udps}) provides guarantees for the \udps-Solver subroutine.
\begin{lemma}
    Let $U$ be a set of $\eta$ jobs such that for any $v \in U,\ |\{ u \in U: u \prec v\}| \le \alpha \delta$. Given input $U$, a set of $\mu$ identical machines, and delay $\delta$, \udps-Solver produces, in polynomial time, a valid \udps\ schedule with makespan at most $3 \alpha\delta\log(\alpha\delta) + (2\eta/\mu)$.
    \label{lem:uniform_jm}
\end{lemma}

\begin{lemma}
    Algorithm~\ref{alg:scheduler_jm} outputs a valid schedule in polynomial time.
\label{lem:alg_valid_polytime}
\end{lemma}
\begin{proof}
    It is easy to see that the algorithm runs in polynomial time, and Lemma~\ref{lem:uniform_jm} entails that precedence constraints are obeyed on each machine. Consider a fixed $v,k,d$ such that $v \in V_{k,d}$. By line~\ref{line:append_jm}, we insert a communication phase of length $\groupdelay k + \groupdelay{\ell}$ before appending the schedule of any set of jobs $V_{k,\ell,d} \cup U_{k,\ell,d}$ on any machine group $\group k$. So, by the time Algorithm~\ref{alg:scheduler_jm} executes any job in $V_{k,d,\ell}$, every job $u$ such that $C_u < d (\groupdelay k + \groupdelay{\ell})$ is available to all machines, including those in group $\group k$. So the only predecessors of $v$ left to execute are those jobs in $U_{k,\ell,d}$. Therefore, all communication constraints are satisfied.
\end{proof}

\begin{lemma}
    If $(C,x,y,z)$ is a rounded solution to \LP a then Algorithm~\ref{alg:scheduler_jm} outputs a schedule with makespan at most $12 \alpha \log( \delay{\max})(K L \Calpha + \delay{\max}(K + L))$.
\label{lem:jm_integer_approximation}
\end{lemma}
\begin{proof}
    Fix any schedule $\sigma$. Note that the schedule produced by the algorithm executes a single job group on a single machine group at a time. Our proof establishes a bound for the total time spent executing a single job group on a single machine group, then sums this bound over all $K$ machine groups and $L$ job groups. 
    
    \begin{claim}
    For any $v,u,k,\ell,d$, if $v \in V_{k,\ell,d}$ and $C_v < C_u + (\groupdelay k + \groupdelay{\ell})$ then $z_{u,v,k,\ell} = 1$. 
    \label{clm:jm_z}
    \end{claim}
    \begin{proof}
    Fix $u,v,k,\ell,d$ such that $v \in V_{k,\ell,d}$ and $C_v < C_u + (\groupdelay k + \groupdelay{\ell})$. By the definition of $V_{k,d}$, $x_{v,k}$ is 1. By constraint~\ref{LPjm:delay}, $C_v \ge C_u + \groupdelay k (1 - z_{u,v,k})$, implying that $z_{u,v,k}$ cannot equal 0.  Since $z_{u,v,k}$ is either 0 or 1, we have $z_{u,v,k} = 1$.
    \end{proof}
    
    \begin{claim}
    For any $k,\ell$, we show 
    \begin{enumerate*}[label=(\alph*)]
        \item $\sum_d |V_{k,\ell,d} \cup U_{k,\ell,d}| \le \Calpha \cdot |\group k|$ and \label{prop:load}
        \item for any $d$ and $v \in V_{k,\ell,d}$, the number of $v$'s predecessors in $V_{k,\ell,d} \cup U_{k,\ell,d}$ is at most $\alpha (\groupdelay k + \groupdelay{\ell})$. \label{prop:phase_preds}
    \end{enumerate*}
    \label{clm:properties}
    \end{claim}
    \begin{proof} 
    Fix $k,\ell$. We first prove \ref{prop:load}. For any $v$ in $V_{k,\ell,d}$ we have $x_{v,k}=1$ by the definition of $V_{k,\ell,d}$.  Consider any $u$ in $U_{k,\ell,d}$. By definition, there exists a $v' \in V_{k,\ell,d}$ such that $x_{v',k} = 1$ and $C_v < C_u + (\groupdelay k + \groupdelay{\ell})$; fix such a $v'$. By claim~\ref{clm:jm_z}, $z_{u,v',k} = 1$. So, by constraint~\ref{LPjm:y>z}, $y_{v,k} = 1$ for every job $v \in V_{k,\ell,d} \cup U_{k,\ell,d}$.  For any $d' \ne d$, $V_{k,\ell,d}$ and $V_{k,\ell,d'}$ are disjoint. So  $\sum_d |V_{k,d} \cup U_{k,d}|$ is at most the right-hand side of constraint~\ref{LPjm:load}, which is at most  $\Calpha \cdot |\group k|$.
    
    We now prove \ref{prop:phase_preds}. Fix $v,d$ such that $v \in V_{k,\ell,d}$. Consider any $u$ in $V_{k,d} \cup  U_{k,d}$ such that $u \prec v$. By definition of $V_{k,\ell,d}$ and $U_{k,\ell,d}$, $C_v < C_u + (\groupdelay k + \groupdelay{\ell})$. By Claim~\ref{clm:jm_z}, $z_{u,v,k} = 1$. The claim then follows from  constraint~(\ref{LPjm:duplicates}). 
    \end{proof}
        
    
    By Lemma~\ref{lem:uniform_jm} and Claim~\ref{clm:properties}\ref{prop:phase_preds}, the  time spent executing jobs in $\jobgroup{\ell}$ on machines in $\group k$ is at most
    \begin{align*}
        \sum_d \left( 3\alpha (\groupdelay k + \groupdelay{\ell}) \log(\alpha (\groupdelay k + \groupdelay{\ell}))   + \frac{2 \cdot |V_{k,\ell,d} \cup U_{k,\ell,d}|}{|\group k|} \right) 
    \end{align*}
    The summation over the first term is at most $\ceil{\Calpha/(\groupdelay k + \groupdelay{\ell})}  3\alpha (\groupdelay k + \groupdelay{\ell}) \log(\alpha (\groupdelay k + \groupdelay{\ell}))$ which is at most $3\Calpha \alpha  \log(\alpha (\groupdelay k + \groupdelay{\ell})) +  3\alpha (\groupdelay k + \groupdelay{\ell}) \log(\alpha (\groupdelay k + \groupdelay{\ell}))$. The summation over the second term is at most $2\Calpha$ by claim~\ref{clm:properties}\ref{prop:load}. Summing over all $K$ machine groups and $L$ job groups, and considering $K,L \le \log \delay{\max}$, the total length of the schedule is at most $12 \alpha \log( \delay{\max})(K L \Calpha + \delay{\max}(K + L))$. 
\end{proof}

\begin{figure}
    \centering
    \includegraphics[width=\textwidth]{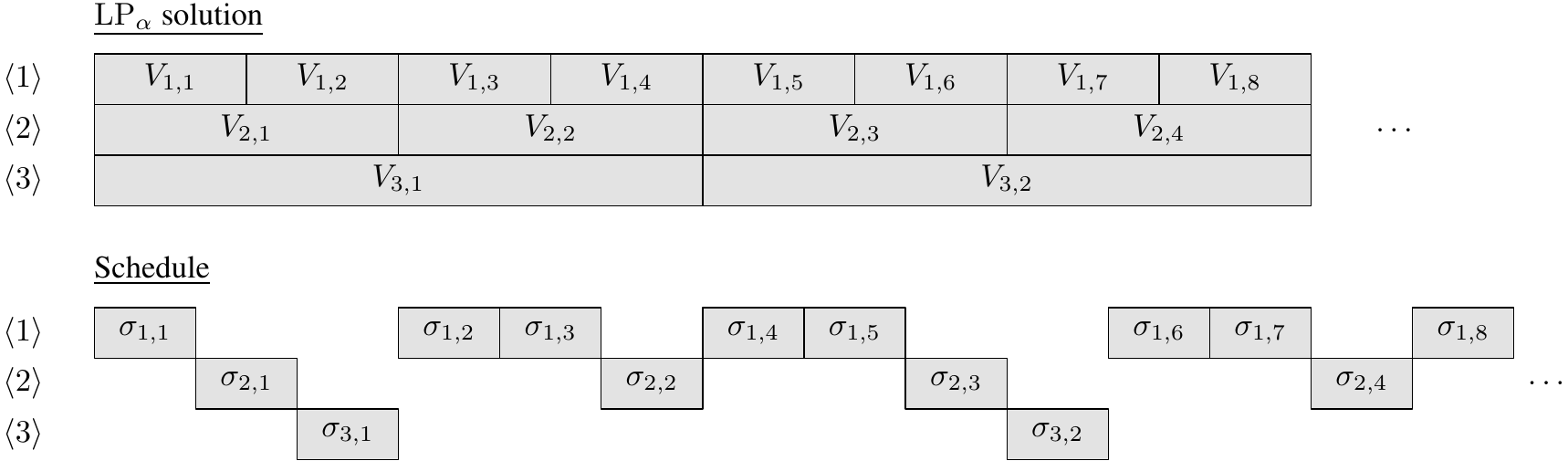}
    \caption{Structure of the schedule produced by Algorithm~\ref{alg:scheduler_jm}. $\sigma_{k,d}$ denotes a schedule of $V_{k,d}$ on the machines in group $\group k$. The algorithm scans the $\text{LP}_{\alpha}$ solution by increasing time (left to right). At the start of each $V_{k,d}$, the algorithm constructs a schedule of the set and appends it to the existing schedule.}
    \label{fig:schedule_structure}
\end{figure}

\begin{theorem}[Job Delays and Machine Delays]
    There exists a polynomial time algorithm to compute a valid machine delays and job precedence delays schedule with makespan $O((\log n)^9(\opt + \delay{\max}))$.
\label{thm:jm_approximation}
\end{theorem}
\begin{proof}
Lemma~\ref{lem:LPKjm_solution} entails that $(C,x,y,z)$ is a valid solution to \LP{2K}. Lemma~\ref{lem:LPKjm_makespan} entails that $\C{2K} \le 4K \cdot \opt$. With $\alpha = 2K$,  Lemma~\ref{lem:jm_integer_approximation}  entails that the makespan of our schedule is at most $12 \alpha \log( \delay{\max})(K L \Calpha + \delay{\max}(K + L)) = 48 (\log \delay{\max})^5 \opt + 24 (\log \delay{\max})^3 \delay{\max}$ for the case with no out-delays. By Lemma~\ref{lem:reduction}, the length of our schedule is $O((\log \delay{\max})^9 (\opt + \delay{\max})$ The theorem is entailed by $\delay{\max} \le n$.  This proves the theorem.
\end{proof}

\begin{corollary}[Machine Delays]
    There exists a polynomial time algorithm to compute a valid machine delays schedule with makespan $O((\log n)^5 \cdot (\opt + \delay{}))$.
\label{cor:machine}
\end{corollary}
\begin{proof}
Lemma~\ref{lem:LPKjm_solution} entails that $(C,x,y,z)$ is a valid solution to \LP{2K}. Lemma~\ref{lem:LPKjm_makespan} entails that $\C{2K} \le 4K \cdot \opt$. With $\alpha = 2K$,  Lemma~\ref{lem:jm_integer_approximation}  entails that the makespan of our schedule is at most $12 \alpha \log( \delay{\max})(K L \Calpha + \delay{\max}(K + L)) = 48 (\log \delay{\max})^5 \opt + 24 (\log \delay{\max})^3 \delay{\max}$ for the case with no out-delays. By Lemma~\ref{cor:reduction_machine}, the length of our schedule is $O((\log \delay{\max})^5 (\opt + \delay{\max})$ The theorem is entailed by $\delay{\max} \le n$. 
\end{proof}


\subsection{Combinatorial Algorithm for Job Delays and Uniform Machine Delays}

The only noncombinatorial subroutine of our algorithm is solving the linear program. In this section, we describe how to combinatorially construct a rounded solution to \LP 1 when machine delays are uniform (i.e.\ for all $i,j,\ \delayin i = \delayout i = \delayin j = \delayout j$), machine speeds are unit, and machine capacities are unit. We let $\delta$ represent the uniform machine delay. By Lemma~\ref{lem:reduction}, we focus on the case where all job out-delays are 0. We let $\delay v = \delta + \delayin v$ for any job $v$.

Since delays, speeds, and capacities are uniform, there is only one machine group: $\group 1$. Set $x_{v,1} = y_{v,1} = 1$ for all $v$. For each job $v$, we define $C_v$ as follows. If $v$ has no predecessors, we set $C_v = 0$. Otherwise, we order $v$'s predecessors such that $C_{u_i} \ge C_{u_{i+1}}$. We define $C_v = \max_{1 \le i \le \delay v} \{C_{u_i} + i\}$. We set $C^* = \max\{ n/m, \max_v\{C_v\}\}$. We set $z_{u,v,1} = 1$ if $u \prec v$ and $C_v - C_u < \delay v$; and set to 0 otherwise. 

\begin{lemma}
    $C^* \le \opt$.
\end{lemma}
\begin{proof}
    Consider an arbitrary schedule in which $t_v$ is the earliest completion time of any job $v$. We show that, for any $v$, $t_v \ge C_v$, which is sufficient to prove the lemma. 
    
    We prove the claim by induction on the number of predecessors of $v$. The claim is trivial if $v$ has no predecessors. Suppose that the claim holds for all of $v$'s predecessors and let $y = \arg\max_{1 \le i \le \delay v} \{C_{u_i} + i\}$. Then $C_v = C_{u_{y}} + y \le t_{u_{y}} + y  \ (\by{IH}) \ = t_{y} + | \{u_x : 0 \le x \le y\} | \le t_y + \delay v$. This entails that all jobs $u_1, \ldots u_y$ must be executed on the same machine as $v$. Now suppose, for the sake of contradiction, that $t_v < C_v$. Then all jobs $u \in \{u_x : 0 \le x < y\}|$ must be executed serially in the time $t_v - t_{u_y} < C_v - t_{u_y} = |\{ u_x : 0 \le x \le y\}|$ which gives us our contradiction.
\end{proof}

\begin{lemma}
    $(C,x,y,z)$ is a rounded solution to \LP 1.
\end{lemma}
\begin{proof}
    It is easy to see that constraints (\ref{LPjm:makespan}, \ref{LPjm:load}, \ref{LPjm:delay},
    \ref{LPjm:precedence}, \ref{LPjm:execution}, \ref{LPjm:C>0}, \ref{LPjm:x>z}, \ref{LPjm:y>x}, \ref{LPjm:y>z} \ref{LPjm:z>0}) are satisfied by the assignment. So we must only show that constraint (\ref{LPjm:duplicates}) is satisfied for fixed $v$. We can see from the definition of $C_v$, that maximum number of predecessors $u$ such that $C_v - C_u < \rho_v + \rho$ is at most $\rho_v + \rho$. This proves the lemma. 
\end{proof}

\begin{lemma}[Combinatorial Algorithm for Job Delays]
    There exists a purely combinatorial, polynomial time algorithm to compute a schedule for Job Delays with makespan $O((\log n)^6(\opt + \max_v\{\rho_v\}))$.
    \label{thm:job}
\end{lemma}
\begin{proof}
    Lemma~\ref{lem:jm_integer_approximation} entails that the length of the schedule is at most $12 (\log \delay{\max})^2(\opt + \delay{\max})$ for the problem with job in-delays. By Lemma~\ref{lem:reduction} we achieve a makespan of $O((\log \delay{\max})^6(\opt + \delay{\max}))$ for job in- and out-delays.  
\end{proof}

\section{Job Delays and Machine Delays with Related Multiprocessor Machines}
\label{sec:general}
In this section, we present an asymptotic approximation algorithm for job delays and machine delays in the presence of variable machine speeds and capacities. In Section~\ref{sec:partition}, we first organize the machines according to their size, speed, and delay.  Section~\ref{sec:LP} presents a linear programming relaxation, assuming that the machines are organized into groups so that all machines in a group have the same size, speed, and delay.   Section~\ref{sec:round} presents a procedure for rounding any fractional solution to the linear program.  Finally, Section~\ref{sec:schedule} presents a combinatorial algorithm that converts an integer solution to the linear program of Section~\ref{sec:LP} to a schedule.  

\paragraph{Model.}
Following Lemma~\ref{lem:reduction}, we focus on the setting where $\delayout i = 0$ for all machines $i$. Since there are no out-delays, we use $\delay i$ to denote the in-delay $\delayin i$ of machine $i$.  

\subsection{Partitioning machines into groups}
\label{sec:partition}

We first define a new set of machines $M'$ such that if $i \in M$, then there is a machine $i' \in M'$ such that $\delayin{i'}, \delayout{i'}$, respectively, equal $\delayin{i}, \delayout{i}$ rounded up to the nearest power of 2, and $m_{i'}, s_{i'}$, respectively, equal $m_{i}, s_{i}$ each rounded down to the nearest power of 2. Similarly, we define a new set of jobs $V'$ such that if $v \in V$ then there is a job $v' \in V'$ with $\delayin{v'}, \delayout{v'}$, respectively, equal $\delayin v, \delayout v$ each rounded up to the nearest power of 2. For the remainder of the paper, we work with jobs $V'$ and machines $M'$. This is justified by the following lemma.

\begin{lemma}
    The optimal makespan over the machine set $M'$ is no more than a factor of 12 greater than the optimal solution over $M$. 
\label{lem:group_rounding}
\end{lemma}
\begin{proof}
    Consider any schedule $\sigma$ on the machine set $M$. We first show that increasing the delay of each machine by a factor of 2 increases the makespan of the schedule by at most a factor of 2. Let $\sigma(v,i)$ be the completion time of $v$ on machine $i$ according to $\sigma$ (and undefined if $v$ is not executed on $i$). We define $\sigma'$ such that $\sigma'(v,i) = 2 \cdot \sigma(v,i)$ (undefined if $\sigma(v,i)$ is undefined). It is easy to see that $\sigma'$ maintains the precedence ordering of jobs, and that the time between the executions of any two jobs has been doubled. Therefore, $\sigma'$ is a valid schedule with all communication delays doubled, and with the makespan doubled.
    
    Next, we show that reducing the size of all machines by a factor of 2 increases the makespan by at most a factor of 3. Let $t_v$ be the completion time of $v$ according to $\sigma$. We define $\sigma'$ as follows. For each machine $i$, we arbitrarily order $i$'s processors $1, \ldots, m_i$. If a job $v$ is executed on processor $\floor{m_i/2}+p$ of machine $i$ at time $t$, then we place it on processor $p$ of machine $i$ at time $t^{new}_v = 2t_v + 1$. If a job $v$ is not executed on an eliminated processor, then it is placed on the same processor and executed at time $t^{new}_v = 2t_v$. It is easy to see that the number of $i$'s processors used is at most half. Suppose that for some pair of jobs $u,v$ we have $t_v - t_u \ge d$. Then $t^{new}_v - t^{new}_u \ge 2t_v - (2t_u+1) = 2(t_v - t_u) - 1 \ge 2d - 1 \ge d$. This shows that all delay and precedence constraints are obeyed and the makespan of $\sigma'$ is at most 3 times the original.
    
    Finally, we show that reducing the speed of each machine by a factor of 2 increases the makespan by at most a factor of 2. Again, by setting $\sigma'(v,i) = 2 \cdot \sigma(v,i)$ (where $\sigma(v,i)$ is defined) we see that the time between any two jobs is increased by a factor of 2 and the makespan is increased by a factor of 2. Therefore, if the speed of each machine is halved, the schedule will remain valid. 
\end{proof}

We partition machines $M'$ into groups according to delay, speed, and size. Group $(\ell_1,\ell_2, \ell_3)$ consists of those machines $i$ such that $2^{\ell_1-1} \le \delay i < 2^{\ell_1}$, $2^{\ell_2} \le m_i \le 2^{\ell_2+ 1}$, and $2^{\ell_3} \le s_i \le 2^{\ell_3+ 1}$. We define a new set of machines $M'$. For every $i \in M$, if $i \in (\ell_1, \ell_2, \ell_3)$, we define $i' \in M'$ such that $\delay i = 2^{\ell_1}$, $m_i = 2^{\ell_2}$, and $s_i = 2^{\ell_3}$. This ensures that every group consists of machines with identical delay, speed, and size, which will simplify the exposition of the algorithm and its proofs. Our reduction to identical machine groups is justified in the following lemma. 

 We order the groups arbitrarily $\group 1, \group 2, \ldots, \group K$ and designate the delay, size, and speed of all machines in group $\group k$ as $\groupdelay k$, $\groupspeed k$, and $\groupsize k$, respectively.  This yields $K = \log \max_k \{ \groupdelay k\} \cdot \log \max_k\{ \groupsize k \} \cdot \log \max_k\{ \groupspeed k\}$ machine groups. Note that if, for any machine $i$, we have $s_i \ge n$, then we can construct a schedule with makespan $1 \le \opt$, for any $\sigma$, by placing all jobs on $i$. So we can assume $\max_k \{ \groupspeed k \} \le n$. Also, suppose $m_i \ge n$ for any $i$. Then, given a schedule $\sigma$, we can construct the exact same schedule on a machine set where $m_i = n$. So we can assume $\max_k \{ \groupsize k \} \le n$. We can also assume that $\max_k \{\bar{\rho}_k\} \le n$ since, if we ever needed a machine with delay greater than $n$ we could schedule everything on a single machine in less time.
 
 We also partiton $V'$ according to job delays. Job $v$ is in group $\jobgroup{\ell}$ if $\delay v = 2^{\ell}$. This yields $\log \{\delay{\max}\} \le \log n$ job groups. 
 
 Therefore, in the most general case, we have $K \le (\log n)^3$ machine groups and $L = \log n$ job groups. Our approximation ratio accrues several factors of $K$ and $L$, so we draw attention to special cases in which $K$ is reduced. The reasoning above implies that if any one of the parameters are uniform across all machines, then we have $(\log n)^2$ groups, and if any two of the parameters are uniform we have $\log n$ groups. We emphasize in particular the case with $\log n$ groups where all machine sizes and speeds are uniform and the delay is variable across the machines.

\subsection{The linear program}
\label{sec:LP}
In this section, we design a linear program \LP{a}, parametrized by  $\alpha \ge 1$, for \mdps.  Following Section~\ref{sec:partition}, we assume that the machines are organized in groups, where each group $\group k$ is composed of machines that have identical sizes, speeds, and delays.

\vspace{-1.5em}
\linearprogram{.63}{.3}{
    &\Calpha \groupspeed k \bar{m}_{k} \cdot |\group k| \ge \sum_{v} y_{v,k} &&\forall k 
    \label{LP:load} 
    \\
    &C_v \ge C_u + (\groupdelay k + \groupdelay{\ell}) (x_{v,k} - z_{u,v,k} ) &&\forall u,v,k,\ell: \label{LP:delay} 
    \\
    &&&\quad u \prec v, v \in \jobgroup{\ell} \notag 
    \\ 
    &\alpha (\groupdelay k + \groupdelay{\ell}) \groupspeed k  \groupsize k \ge \sum_u z_{u,v,k} &&\forall v,k,\ell: v \in \jobgroup{\ell}
    \label{LP*:load}
    \\
    &C^{v,k}_u \ge C^{v,k}_{u'} +  \frac{z_{u,v,k}}{\groupspeed k} &&\forall u,u',v, k: u' \prec u
    \label{LP*:precedence}
    \\
    &\alpha (\groupdelay k + \groupdelay{\ell}) \ge C^{v,k}_u &&\forall u, v, k, \ell: v \in \jobgroup{\ell}
    \label{LP*:makespan}
    \\
    &C_v \ge C_u + \sum_k \frac{x_{v,k}}{\groupspeed k} &&\forall u,v: u \prec v
    \label{LP:precedence} 
}{
    &\Calpha \ge C_v &&\forall v 
    \label{LP:makespan} 
    \\[.1em]
    &\sum_{k} x_{v,k} = 1 &&\forall v
    \label{LP:execution}
    \\[.1em]
    &C_v \ge 0 && \forall v
    \label{LP:mincompletion} 
    \\[.1em]
    &C^{v,k}_u \ge  0 &&\forall v,u, k 
    \label{LP*:nonneg} 
    \\[.1em]
    &x_{v,k} \ge z_{u,v,k} &&\forall u,v, k
    \label{LP:minexecution} 
    \\[.1em]
    &y_{v,k} \ge x_{v,k} &&\forall v,k
    \label{LP:duploadx} 
    \\[.1em]
    &y_{u,k} \ge z_{u,v,k} &&\forall u,v,k
    \label{LP:duploadz} 
    \\[.1em]
    &z_{u,v,k} \ge 0 &&\forall u,v,k
    \label{LP:dupbounds} 
}

\begin{lemma}{\textbf{(\LP 1 is a valid relaxation)}}
    The minimum of $\C 1$ is at most $\opt$.
\label{lem:LP_valid}
\end{lemma}
\begin{proof}
    Consider an arbitrary schedule $\sigma$ with makespan $\Csigma$. We assume the schedule is \textit{minimal} in the sense that a job is executed at most once on any machine and no job is executed on any machine after the result of that job is available to the machine. It is easy to see that if $\sigma$ is not minimal, there exists a minimal schedule whose makespan is no more than $\sigma$'s. We give a solution to the LP in which $\Calpha = \Csigma$.
    
    \smallskip
    \BfPara{LP solution} Set $\Calpha = \Csigma$. For each job $v$, set $C_v$ to be the earliest completion time of $v$ in $\sigma$. Set $x_{v,k} = 1$ if $v$ first completes on a machine in group $\group k$ (0 otherwise). For $u,v,k$, set $z_{u,v,k} = 1$ if $u \prec v$, $x_{v,k} = 1$, $v \in \jobgroup{\ell}$, and $C_v - C_u < \groupdelay k + \groupdelay{\ell}$ (0 otherwise). Set $y_{v,k} = \max\{x_{v,k}, \max_{u,\ell} \{ z_{v,u,k,\ell} \} \}$.
    Fix a job $v$ and group $\group k$. Let $U = \{ u : z_{u,v,k} = 1\}$. If $U = \varnothing$, then set $C^{v,k}_u = 0$ for all $u$. Otherwise, suppose $u \in U$. Then, by assignment of $z_{u,v,k}$, the earliest completion of $v$ is on some machine $i \in \group{k}$. Since the schedule is minimal, we can infer that exactly one copy of $u$ is executed on $i$ as well. Let $t^v_u$ be the completion time of $u$ on $i$. Set $C^{v,k}_u = \groupdelay k + \groupdelay{\ell} - (C_v - t^v_u)$. Now suppose $u \not\in U$. If there is no $u' \in U$ such that $u' \prec u$, then set $C^{v,k}_u = 0$. Otherwise, set $C^{v,k}_u = \max_{u' \prec u} \{C^{v,k}_{u'}\}$. 
    
    \smallskip
    \BfPara{Feasibility} We now establish that the solution defined is feasible. Constraints (\ref{LP:precedence} - \ref{LP:dupbounds}) are easy to verify.  Consider constraint~(\ref{LP:load}) for fixed group $\group k$. $\sum_v y_{v,k}$ is upper bound by the total load $L$ on $\group k$. The constraint follows from $\Calpha \ge \Csigma \ge L / (|\group k| \cdot \groupsize k \groupspeed k)$. Consider constraint~(\ref{LP:delay}) for fixed $u,v,k,\ell$ where $u \prec v$ and $v \in \jobgroup{\ell}$. Let $X = x_{v,k}$ and let $Z = z_{u,v,k}$. If $x_{v,k} - z_{u,v,k} \le 0$ then the constraint follows from constraint~(\ref{LP:precedence}). Otherwise, $x_{v,k} = 1$ and $z_{u,v,k} = 0$. By assignment of $z_{u,v,k}$ we can infer that $C_v - C_u \ge \groupdelay k + \groupdelay{\ell}$, which shows the constraint is satisfied.
    
    To show that constraints (\ref{LP*:load} - \ref{LP*:makespan}) are satisfied, consider fixed $v,k,\ell$ and let $U = \{ u: z_{u,v,k} = 1 \}$. If $v$ is not executed on any machine in $\group{k}$ (i.e.\ $x_{v,k} = 0$) then $U = \varnothing$ and the constraints are trivially satisfied by the fact that $C^{v,k}_{u} = 0$ for all $u$. So suppose that $x_{v,k} = 1$.
    \begin{claim}
    The completion times given by $\{C^{v,k}_u\}_{u \in U}$ give a valid schedule of $U$ on any machine in $\group{k}$. 
     \end{claim} 
   \begin{proof}
        Suppose the schedule were not valid for a given machine $i \in \group k$. By assignment of the $C^{v,k}_u$ variables, there are at most $\groupsize k$ jobs scheduled on $i$ in any time step. So the completion times fail to give a schedule only because there is some $u \prec u'$ such that $t^v_u \ge t^v_{u'}$. By the fact that the schedule is minimal, the result of $u$ is not available to $i$ by time $t^v_u$, and so is not available to $i$ by time $t^v_{u'}$. Therefore, the original schedule is not valid, which contradicts our supposition.
    \end{proof}
    
    We proceed to show that (\ref{LP*:load} - \ref{LP*:makespan}) are satisfied using the same fixed $v,k$ and the same definition of $U$.  
    First consider constraint~(\ref{LP*:makespan}) for fixed $u$. Let $u^* = \arg\max_{u' \in U} \{C^{v,k}_{u'} \}$. By assignment, we have $C^{v,k}_{u^*} \ge C^{v,k}_u$. By definition, we have $C_v = t^v_v$ and $t^v_v > t^v_{u^*}$. Therefore, $C^{v,k}_u \le C^{v,k}_{u^*} = (\groupdelay k + \groupdelay{\ell}) - (C_v - t^v_{u^*}) = (\groupdelay k + \groupdelay{\ell}) - (t^v_v - t^v_{u^*}) \le (\groupdelay k + \groupdelay{\ell})$, which shows the constraint is satisfied.
    Consider constraint~(\ref{LP*:load}). Let $\widehat{C}$ be the optimal max completion time of scheduling just $U$ on a machine in $\group{k}$. Then
    $\sum_u z_{u,v,k} = |U| \le \widehat{C} \cdot \groupspeed k \groupsize k \le \max_{u \in U} \{C^{v,k}_u\} \cdot \groupspeed k \groupsize k$
    by the claim above. By constraint~(\ref{LP*:makespan}), we have that $\max_{u \in U} \{C^{v,k}_u\} \le \groupdelay k + \groupdelay{\ell}$, which entails that the constraint is satisfied.
    Consider constraint~(\ref{LP*:precedence}) for fixed $u,u'$ such that $u' \prec u$. If $z_{u,v,k} = 0$, then the constraint holds by the fact that $C^{v,k}_u = \max_{\hat u \prec u} \{C^{v,k}_{\hat{u}}\}$. If $z_{u,v,k} = 1$ then the constraint holds by the claim above.
\end{proof}


\subsection{Deriving an rounded solution to the linear program}
\label{sec:round}

Let \LP 1 be defined over machine groups $\group 1, \group 2, \ldots, \group K$ and job groups $\jobgroup 1, \jobgroup 2, \jobgroup L$. Given a solution $(\hat C,\hat x, \hat y, \hat z)$ to \LP 1, we construct an integer solution $(C,x,y,z)$ to \LP{2K} as follows. For each $v,k$, set $x_{v,k} = 1$ if $k = \max_{k'}\{\hat x_{v,k'}\}$ (if there is more than one maximizing $k$, arbitrarily select one); set to 0 otherwise. Set $z_{u,v,k} = 1$ if $x_{v,k} = 1$ and $\hat z_{u,v,k} \ge 1/(2K)$; set to 0 otherwise. For all $u,k$, $y_{v,k} = \max\{ x_{v,k}, \max_u \{z_{v,u,k}\}$. Set $C_v = 2K \cdot \hat C_v$. Set $C^{v,k}_u =  2K \cdot \hat C^{v,k}_u$. Set $\C{2K} = 2K \cdot \hat{\C 1}$. 

\begin{lemma}
    If $(\hat C, \hat x, \hat y, \hat z)$ is a valid solution to \LP 1, then  $(C, x,y,z)$ is a valid solution to \LP{2K}. 
\label{lem:LPK_solution}
\end{lemma}
\begin{proof}
    By constraint~\labelcref{LP:execution}, $\sum_{k} \hat x_{v,k}$ is at least 1, so $\max_{k} \{\hat x_{v,k}\}$ is at least $1/K$.  Therefore, $x_{v,k} \le K\hat{x}_{v,k}$ for all $v$ and $k$. Also, $z_{u,v,k} \le 2K \hat z_{u,v,k}$ for any $u,v,k$ by definition. By the setting of $C_v$ for all $v$, $C_{u}^{v,k}$ for all $u, v, k$, $y_{v,k}$ for all $v,k$, and $\C{2K}$, it follows that constraints~(\ref{LP*:load}-\ref{LP:dupbounds}) of \LP{1} imply the respective constraints of \LP{2K}.
    We first establish constraint~\labelcref{LP:load}. For any fixed group $\group k$,
    \begin{align*}
        2K \hat{C}_1 \cdot |\group k| &\ge 2K \sum_v \hat y_{v,k} = 2K \sum_v \max\{ \hat x_{v,k}, \max_u \{ \hat z_{v,u,k}\} \}  &&\by{constraints~(\ref{LP:load}, \ref{LP:duploadz}, \ref{LP:duploadx}) of \LP 1} \\
        &\ge 2K \sum_v \frac{ x_{v,k} + \max_u \{ z_{v,u,k} \}}{2K} \ge \sum_v y_{v,k} &&\by{definition of $y_{v,k}$} 
    \end{align*}
    which entails constraint~(\ref{LP:load}) by $\C{2K} = 2K \hat{\C 1}$.
    It remains to establish constraint~(\ref{LP:delay}) for fixed $u,v,k,\ell$.  We consider two cases.  If $\hat{x}_{v,k} < 1/K$, then $x_{v,k} = 0$, so the constraint is trivially true in \LP{2K}.  Otherwise, $x_{v,k} - z_{u,v,k}$ equals $1 - z_{u,v,k}$, which is 1 only if $x_{v,k} - z_{u,v,k}$ is at least $1/(2K)$.  This establishes constraint~(\ref{LP:delay}) of \LP{2K} and completes the proof of the lemma.
\end{proof}

\begin{lemma}
    $\C{2K} \le 24 K \cdot \opt$.
\label{lem:LPK_makespan}
\end{lemma}
\begin{proof}
    Lemma~\ref{lem:group_rounding} shows that our grouping of machines does not increase the value of the LP by more than a factor of 12. Therefore, by Lemma~\ref{lem:LP_valid}, $\C{2K} = 2K \hat{\C 1} \le 24 K \cdot \opt$. 
\end{proof}

\subsection{Computing a schedule given an integer solution to the LP}
\label{sec:schedule}

Suppose we are given a partition of $M$ into $K$ groups such that group $\group k$ is composed of identical machines (i.e. for all $i,j \in \group k$, $s_i = s_j$, $m_i = m_j$, and $\rho_i = \rho_j$). We are also given a partition of $V$ into $L$ groups such that $\jobgroup{\ell}$ is composed of jobs with identical delays (i.e.\ for all $u,v \in \jobgroup{\ell}, \delay u = \delay v$). Finally, we are given a valid integer solution $(C,x,y,z)$ to \LP a defined over machine groups $\group 1, \ldots, \group K$ and job groups $\jobgroup{1}, \jobgroup{2}, \ldots, \jobgroup{L}$. In this section, we show that we can construct a schedule that achieves an approximation for job and machine delays with related multiprocessor machines in terms of $\alpha$, $K$, and $L$. The combinatorial subroutine that constructs the schedule is defined in Algorithm~\ref{alg:scheduler}.

We define the Uniform Delay Precedence-constrained Scheduling (\udps) problem as in Section~\ref{sec:job_machine_delays}.

\begin{lemma}
    Let $U$ be a set of jobs such that for any $v \in U$ the number of predecessors of $v$ in $U$ (i.e., $|\{ u \prec v\} \cap U|$) is at most $\alpha \delay{} s \mu$, and the longest chain in $U$ has length at most $\alpha \delay{}$. Then given as input jobs $U$, $m'$ machines with speed $s$ and size $\mu$, and delay $\delay{}$, \udps-Solver produces, in polynomial time, a valid schedule with makespan less than $3\log(\delay{} s \mu)\alpha \delay{} + \frac{2 \cdot |U|}{m' s \mu}$.
\label{lem:uniform}
\end{lemma}

For a proof of Lemma~\ref{lem:uniform}, see Section~\ref{sec:uniform}. The subroutine is the same as in Section~\ref{sec:job_machine_delays}, but using the job and machine groups as defined in this section. Lemma~\ref{lem:alg_valid_polytime} establishes that this subroutine outputs a valid schedule in polynomial time.



\begin{lemma}
    If $(C,x,y,z)$ is a rounded solution to \LP a then Algorithm~\ref{alg:scheduler_jm} outputs a schedule with makespan at most $12 \alpha \log( \delay{\max})(K L \Calpha + \delay{\max}(K + L))$.
\label{lem:integer_approximation}
\end{lemma}
\begin{proof} 
    The lemma follows from the following claims.
    
    \begin{claim}
    \label{clm:z}
    For any $u,k,\ell,d$, if $u$ is in  $U_{k,\ell,d}$, then for some $v$ in $V_{k,\ell,d}$ we have $z_{u,v,k} = 1$. 
    \end{claim}
    \begin{proof}
    Fix a job $u$.  If $u$ is in $U_{k,\ell,d}$, then by the definition of $U_{k,\ell,d}$, there exists a job $v$ in $V_{k,\ell,d}$ such that $u \prec v$ and $C_u$ is in $[T, T + \groupdelay k + \groupdelay{\ell})$, where $T = d(\groupdelay k + \groupdelay{\ell})$.  By the definition of $V_{k,\ell,d}$, $x_{v,k}$ is 1 and $C_v$ is also in $[T, T + \groupdelay k + \groupdelay{\ell})$.  Therefore, $C_v - C_u < \groupdelay k + \groupdelay{\ell}$.  By constraint~\labelcref{LP:delay}, $C_v \ge C_u + (\groupdelay k + \groupdelay{\ell})(1 - z_{u,v,k})$, implying that $z_{u,v,k}$ cannot equal 0.  Since $z_{u,v,k}$ is either 0 or 1, we have $z_{u,v,k} = 1$.
    \end{proof}
    
    \begin{claim}
    \label{clm:chain}
    For any $v,k,\ell,d$, if $v \in V_{k,\ell,d}$ then $|\{u \prec v: u \in V_{k,\ell,d} \cup U_{k,\ell,d}\}|$ is at most $\alpha (\groupdelay k + \groupdelay{\ell}) \groupsize k \groupspeed k$ and the longest chain in $U_{k,\ell,d}$ has length at most $\alpha (\groupdelay k + \groupdelay{\ell}) \groupspeed k$.
    \end{claim}
    \begin{proof} 
    Fix $v$ and consider any $u$ in $V_{k,\ell,d} \cup  U_{k,\ell,d}$ such that $u \prec v$.  
    By Claim~\ref{clm:z}, $z_{u,v,k} = 1$.  Therefore, $|\{u \prec v: u \in V_{k,\ell,d} \cup U_{k,\ell,d}\}|$ equals the right-hand side of constraint~\labelcref{LP*:load}, and hence at most $\alpha \bar{\rho}_k \groupsize k \groupspeed k$.  Let $u_1 \prec u_2 \prec \cdots \prec u_\ell$ denote a longest chain in $U_{k,\ell,d}$.  By Claim~\ref{clm:z}, $z_{u_i,v,k}$ equals 1 for $1 \le i \le \ell$.  By constraint~\labelcref{LP*:nonneg}, $C_{u_1}^{v,k} \ge 0$.  By a repeated application of constraint~\labelcref{LP*:precedence}, we obtain $C_{u_\ell}^{v,k} \le \ell$.  By constraint~\labelcref{LP*:makespan}, $\ell \le \alpha \bar{\rho}_k \groupspeed k$, and hence the longest chain in $U_{k,\ell,d}$ has length at most $\alpha \bar{\rho}_k \groupspeed k$.
    \end{proof}
        
    \begin{claim}
    \label{clm:D}
    For any $k$, $|\bigcup_d V_{k,\ell,d} \cup U_{k,\ell,d}| \le D \cdot |\group k| \cdot \groupsize k \groupspeed k$.
    \end{claim} 
    \begin{proof} 
    Fix a $k$.  For any $v$ in $V_{k,\ell,d}$ we have $x_{v,k}=1$ by the definition of $V_{k,\ell,d}$.  Consider any $u$ in $U_{k,\ell,d}$.  By Claim~\ref{clm:z}, there exists a $v$ in $V_{k,\ell,d}$ such that $z_{u,v,k}$ equals 1; fix such a $v$.  By constraint~\labelcref{LP:duploadz}, $y_{u,k}$ equals 1.  Thus, $|\bigcup_d V_{k,\ell,d} \cup U_{k,\ell,d}|$ is at most the right-hand side of constraint~\labelcref{LP:load}, which is at most  $\Calpha \cdot |\group k| \cdot \groupsize k \groupspeed k$.
    \end{proof}
    
    So, by Lemma~\ref{lem:uniform} and Claim~\ref{clm:chain}, the total time spent executing jobs on a single group is upper bounded by
    \begin{align*}
        \sum_d \left( 3\alpha (\groupdelay k + \groupdelay{\ell}) \log(\alpha (\groupdelay k + \groupdelay{\ell}))   + \frac{2 \cdot |V_{k,\ell,d} \cup U_{k,\ell,d}|}{|\group k|} \right) 
    \end{align*}
    The summation over the first term is at most $\ceil{\Calpha/(\groupdelay k + \groupdelay{\ell})}  3\alpha (\groupdelay k + \groupdelay{\ell}) \log(\alpha (\groupdelay k + \groupdelay{\ell}))$ which is at most $3\Calpha \alpha  \log(\alpha (\groupdelay k + \groupdelay{\ell})) +  3\alpha (\groupdelay k + \groupdelay{\ell}) \log(\alpha (\groupdelay k + \groupdelay{\ell}))$. The summation over the second term is at most $2\Calpha$ by claim~\ref{clm:properties}\ref{prop:load}. Summing over all $K$ machine groups and $L$ job groups, and considering $K,L \le \log \delay{\max}$, the total length of the schedule is at most $12 \alpha \log( \delay{\max})( \Calpha K L + \log (\delay{\max})(K + L))$. 
\end{proof}

\begin{theorem}
    There exists a polynomial time algorithm that produces a schedule with makespan $O((\log n)^{15}(\opt + \delay{\max})$.
\end{theorem}
\begin{proof}
    Lemma~\ref{lem:LPK_solution} entails that $(C,x,y,z)$ is a valid solution to \LP{2K}. By Lemma~\ref{lem:integer_approximation}, the makespan of the constructed schedule is $12(2K)(\log \delay{\max})(KL(2K \cdot \opt) + \delay{\max}(K + L))$. The theorem is entailed by $K \le (\log n)^3$, $L \le \log n$, and $\delay{\max} \le n$. 
\end{proof}

For the most general case with variable delays, sizes, and speeds, we have $K = (\log n)^3$ groups, yielding an $O((\log n)^{10} \cdot (\makespan{\sigma} + \delay{\sigma}))$ approximation for any schedule $\sigma$. For special cases in which any one parameter is uniform across all machines, we have $K = (\log n)^2$ groups, yielding an $O((\log n)^{7} \cdot (\makespan{\sigma} + \delay{\sigma}))$ approximation for any schedule $\sigma$. For special cases in which any two parameters are uniform across all machines, we have $K = (\log n)$ groups, yielding an $O((\log n)^{4} \cdot (\makespan{\sigma} + \delay{\sigma}))$ approximation for any schedule $\sigma$. 


\section{No-Duplication Schedules under Job and Machine Delays}
\label{sec:no-dup}
In Section~\ref{sec:mdps-no-dup}, we present an asymptotic polylogarithmic approximation for scheduling related machines with multiprocessors under machine delays and job delays, without duplication (Lemma~\ref{lem:no-dup}).  In Section~\ref{sec:mdps-symmetric-no-dup}, we show that when the delays are symmetric (i.e., in- and out-delays are identical), we can in fact derive true polylogarithmic approximations for no-duplication schedules (Lemma~\ref{lem:no-dup-symmetric}).  Theorem~\ref{thm:no-dup} follows immediately from Lemmas~\ref{lem:no-dup} and~\ref{lem:no-dup-symmetric}.
\subsection{Asymptotic Approximation}
\label{sec:mdps-no-dup}
Our algorithm for no-duplication schedules first runs Algorithm~\ref{alg:scheduler_jm} for the problem, which produces a schedule (with possible duplications) having a structure illustrated in Figure~\ref{fig:schedule_structure}.  By line 7 of Algorithm~\ref{alg:scheduler_jm}, the schedule $\sigma_{k,\ell,d}$ for each $V_{k,\ell,d}$ includes a delay of $\groupdelay k + \groupdelay{\ell}$ at the start.  We next 
convert the schedule $\sigma_{k,\ell,d}$, for each $V_{k,\ell,d}$, to a no-duplication schedule using the following lemma from \cite{maiti_etal.commdelaydup_FOCS.20}.  
\begin{lemma}[Restatement of Theorem 3 from \cite{maiti_etal.commdelaydup_FOCS.20}]
  Given a \udps\ instance where the delay is $\rho$ and a schedule of length $D \geq \rho$, there is a polynomial time algorithm that  computes a no-duplication schedule for the instance with length $O((\log m )(\log n)^{2}D)$.
  \label{lemma:focs}
\end{lemma}
Let $\hat{\sigma}_{v,k,\ell}$ denote the no-duplication schedule thus computed for $V_{k,\ell,d}$.  We concatenate the $\hat{\sigma}_{v,k,\ell}$ in the order specified in Figure~\ref{fig:schedule_structure}.  By construction, there is no duplication of jobs within any $\hat{\sigma}_{v,k,\ell}$.  To ensure that the final schedule has no duplication, we keep the first occurrence of each job in the schedule, and prune all duplicate occurrences.  We now show that no precedence constraints are violated.
\begin{lemma}
  In the above algorithm, all precedence constraints are satisfied.
  \label{lem:de-dup-constraints}
\end{lemma}

\begin{proof}
Suppose $u \prec v$. Pruning of a duplicate copy of $u$ can only cause a precedence violation if the copy of $u$ that $v$ waits for is pruned, so we assume $u$ is pruned. For each copy of $v$ scheduled on some machine $i$, either some copy of $u$ is also scheduled on $i$, or it needs to wait $\delay i + \delay v$ time to wait for those completed on other machines to be transmitted. In the former case, $u$ will only be pruned if it is not the first occurrence. This can only happen if the first occurrence of $u$ is in an earlier $V_{k,\ell, d}$, since our subroutine guarantees that there is no duplication within a $V_{k,\ell, d}$. Due to our in-delay model, as long as a later occurrence of $u$ can be transmitted in time, so can all its previous occurrences. Therefore a copy of $u$ will be available to $v$. In the latter case, either $u$ is within the same $V_{k,\ell, d}$, or is in a previous one. In either case, the earlier occurrences of $u$ have to be in earlier $V_{k,\ell, d}$, and will be available to $v$.  So the pruning of this copy of $u$ will not cause a precedence violation.
\end{proof}

\begin{lemma}
There exists a polynomial time algorithm that produces a no-duplication schedule whose makespan is at most $\polylog(n) (\opt + \delay{\max})$.
\label{lem:no-dup}
\end{lemma}
\begin{proof}
Lemma~\ref{lem:de-dup-constraints} establishes the feasibility of the schedule. By definition, each job is scheduled exactly once. It remains to bound the approximation ratio.  Let $\tau_{k, d}$ be the makespan of the schedule of $V_{k, d}$ given by the subroutine we use on line 6 of \Cref{alg:scheduler}, and $\tau_{k, d}'$ be the makespan of the de-duplicated version. Then \Cref{lemma:focs} shows that $\tau_{k, d}' \leq O((\log n_{k, d})^{2}(\log m)(\tau_{k, d} + \rho_{k}))$, where $n_{k, d}$ is the number of jobs in $V_{k, d}$, which is upper bounded by $n$.  There are at most $n$ jobs to schedule, so having more than $n$ machines does not help us, therefore we can assume $m \leq n$ without loss of generality. \Cref{alg:uniform} guarantees that there is at least an in-delay at the start of $\tau_{k, d}$, so $\tau_{k, d} \geq \rho_{k}$, therefore $\tau_{k, d}' \leq O((\log n_{k, d})^{3}\tau_{k, d})$. Therefore, the total makespan
\[\sum_{k, d}\tau_{k, d}' \leq \sum_{k, d}O((\log n)^{3}\tau_{k, d})  \leq O((\log n)^{3}) \cdot \sum_{k, d}\tau_{k, d} = O(\text{polylog}(n) \makespan{\sigma} + \delay{\sigma}))\]
\end{proof}

\subsection{Symmetric Job and Machine Delays}
\label{sec:mdps-symmetric-no-dup}
In this section, we establish a true polylogarithmic approximation for the no-duplication model if the delays are symmetric, i.e., $\delayin i = \delayout i$ for all machines $i$ and $\delayin v = \delayout v$ for all jobs $v$.  The additive term of delay in Theorems~\ref{thm:smdjd-general} and Lemma~\ref{lem:no-dup} comes from \Cref{alg:uniform}, where to ensure all previously scheduled jobs are available, we insert a suitable delay at the front.  This can be a cost too high if, for instance, the optimal solution does not incur any communication 
into a particular job or machine group.  To overcome this difficulty, we develop a method that, given a threshold $T$, determines which groups participate in a schedule of length $T$ without any communication.  Formally, we establish the following lemma which, together with a binary search on $T$, yields the second part of~\Cref{thm:no-dup}.
\begin{lemma}
\label{lem:no-dup-symmetric}
There exists a polynomial time algorithm such that given $T$, either gives a schedule of length $O(\polylog (n)T)$, or correctly asserts that $\opt > T$.
\end{lemma}

Let $S$ (resp., $S'$) be the set of machines that have delay at most (resp., larger than) $T$.  
Let $G$ denote the DAG defined by the set of jobs and its precedence constraints.  Let $G'$ denote the undirected graph obtained from $G$ by removing directions from all edges.  Under symmetric delays,  we can afford no communication in either direction with machines in $S'$, so in the absence of duplication both \opt\  and our algorithm will only schedule full connected components of $G'$ on $S'$.  A priori, it is unclear how to distribute the connected components among $S$ and $S'$.  We revise the LP framework of Section~\ref{sec:general} so that it guides this distribution while yielding a fractional placement without duplication on $S'$ and a fractional placement with duplication on $S$.  We then show how to design a full no-duplication schedule from these fractional placements.

For each connected component $d$ and each machine $i\in S'$, we define an LP variable $X_{d, i}$, indicating whether the component is scheduled on machine $i$. Therefore for the completion of jobs, we replace constraint~\labelcref{LP:execution} with
\begin{equation}
  \label{LP:execution_dedup_simple}
\sum_{k} x_{v, k} + \sum_{i\in S'} X_{d, i} = 1, \quad \forall d, \forall v\in d
\end{equation}
Use $w(d)$ to denote the total number of jobs in connected component $d$, and $L(d)$ the length of its critical path. We also need to add the following constraints:
\begin{align}
  &D\geq \sum_{d}\frac{X_{d, i}\cdot w(d)}{m_{i}s_{i}}, &&\forall i \label{LP:makespan_Xdi}\\
  &0 \leq X_{d, i} \leq 1, &&\forall d, i \label{LP:Xdi}\\
  &X_{d, i} = 0, &&\forall d, \text{ if } l(d) > s_{i}\cdot T \label{LP:Xdi:critical}\\
  &X_{d, i} = 0, &&\forall d, \text{ if } w(d) > m_{i}\cdot s_{i}\cdot T \label{LP:Xdi:load}
\end{align}

After solving the LP, if $D > T$, we know \textsf{OPT} cannot have a schedule less than $T$, since a no-duplication schedule is also a valid duplication schedule, and satisfies the revised LP. If $D \leq T$, we proceed to round and give a schedule. Let $\delta_{d} = \sum_{i}X_{d, i}$. For each component $d$ such that $\delta_{d} \geq \frac{1}{2}$, we assign them to the long delay machines $S'$. To be more precise, we increase their $X_{d, i}$ variables by a factor of $1/\delta_{d}$, and set the corresponding $x_{v, k}, z_{u, v, k}$ variables to be zero. Since $\delta_d \geq 1/2$, we have the corresponding constraints.
\begin{align}
  &\sum_{i\in S'} X_{d, i} = 1, &&\forall  d \label{LP:execution_dedup}\\
  &2 D\geq \sum_{d}\frac{X_{d, i}\cdot w(d)}{m_{i}s_{i}}, &&\forall i \label{LP:makespan_Xdi_rounded}
\end{align}

For all other components with $\delta_{d} < \frac{1}{2}$, set their $X_{d, i} = 0$. For every job $v$ in component $d$, scale its corresponding $x_{v, k}, y_{v, k}, z_{u, v, k}, C_{v}, C_{u}^{v, k}$ up, by a factor of $1/(1-\delta_{d})$. In addition, we also replace $\alpha$ with $2\alpha$, $D$ with $2D$, since $2$ is an upper bound of $1/(1 - \delta_{d})$. It is not hard to check that all constraints~(\ref{LP:load}, \ref{LP:delay}, \ref{LP*:load}, \ref{LP*:precedence}, \ref{LP*:nonneg}, \ref{LP*:makespan}, \ref{LP:makespan}, \ref{LP:precedence}, \ref{LP:mincompletion}, \ref{LP:minexecution}, \ref{LP:duploadx}, \ref{LP:duploadz}, \ref{LP:dupbounds})
hold, and constraint~\ref{LP:execution_dedup_simple} degrades to constraint~\labelcref{LP:execution}.  We have thus established that $(C, x, y, z)$ as defined above satisfies the LP.

Now we need to further the LP variables into a solution.  So we only need to round the non-zero $X_{d,i}$ into an integral solution and get a schedule of connected components (and their jobs) on $S'$ within time polylog of $T$, and we actually upper bound it with $2T + 2D = O(T)$. We first focus on the load $w(d)$. Viewing each component as a job, the problem of assigning components to machines is exactly minimizing makespan while scheduling general length jobs on related machines. For this problem, we use a classic result due to Lenstra, Shmoys, and Tardos, restated here for completeness~\cite{lenstra1990approximation}. 

\begin{lemma}[Restatement of Theorem 1 from \cite{lenstra1990approximation}]
  Let $P = (p_{i, j}) \in \mathbb{Z}_{+}^{m\times n}, (d_{1}, \dots, d_{m})\in \mathbb{Z}_{+}^{m}$, and $t\in \mathbb{Z}_{+}$.
  Let $J_{i}(t)$ denote the set of jobs that requre no more than $t$ time units on machine $i$, and let $M_{j}(t)$ denote the set of machines that can process job $j$ in no more than $t$ time units. Consider a decision version of scheduling problem where for each machine $i$ there is a deadline $d_{i}$ and where we are further constrained to schedule jobs so that each uses processing time at most $t$; we wish to decide if there is a feasible schedule. If the linear program
  \begin{align*}
    \sum_{i\in M_{j}(t)}x_{ij} &= 1, && j = 1, \dots, n\\
    \sum_{j\in J_{i}(t)}p_{ij}x_{ij} &\leq d_{i}, &&  i = 1, \dots, m\\
    x_{ij} &\geq 0, && \forall j\in J_{i}(t), i = 1, \dots, m
  \end{align*}
  has a feasible solution, then any vertex $\tilde{x}$ of this polytope can be rounded to a feasible solution $\bar{x}$ of the integer program
  \begin{align*}
    \sum_{i\in M_{j}(t)}x_{ij} &= 1, && j = 1, \dots, n\\
    \sum_{j\in J_{i}(t)}p_{ij}x_{ij} &\leq d_{i} + t, &&  i = 1, \dots, m\\
    x_{ij} &\in \{0, 1\}, && \forall j\in J_{i}(t), i = 1, \dots, m
  \end{align*}
\label{lem:Lenstra}
\end{lemma}

If we view components as items and set $p_{i,d} = w_{d}/(m_{i}s_{i})$, the constraints of Lemma~\ref{lem:Lenstra} are the same as  constraints~\labelcref{LP:execution_dedup,LP:makespan_Xdi_rounded}. Therefore, we can round them into an integer solution such that
\begin{equation}
  \sum_{d}\frac{X_{d, i}\cdot w(d)}{m_{i}s_{i}} \leq T + \frac{1}{\delta_{d}}\cdot D, \quad \forall i \label{LP:makespan_Xdi_rounded_integer}
\end{equation}
In other words, for machine $i\in S'$, the total number of jobs on it is at most $m_{i}s_{i}(T + 2D)$. Since each DAG has critical path length $T$, using Graham's list scheduling, we can find a schedule for machine $i$ in time $\frac{m_{i}s_{i}(T + 2D)}{m_{i}s_{i}} + T = 2T + 2D$.

We now handle the jobs assigned to machines in $S$. We define the set $\Vlong = \{v: \delay v > T$ and $v$ not assigned to $S'\}$. For each $v \in \Vlong$, we define $W_v = \{u : u \prec v$ or $v \prec u\}$. We then merge sets that have overlapping elements: define $\tilde{W}_v =  \bigcup_{v': W_v \cap W_{v'} \ne \varnothing} W_{v'}$. Clearly each set $\tilde{W}_v$ must be scheduled on a single machine to construct a schedule of length less than $T$. We order machines in $S$ by non-increasing capacity: if $i > j$ then $s_i m_i \ge s_j m_j$. For each machine $i \in S$, we define the set $V_i = \{v: |\tilde{W}_v| > T s_i m_i$ or the length of the longest chain in $\tilde{W}_v$ has length greater than $T s_i\}$. 

\begin{lemma}
    In any schedule with length less than $T$, jobs in $V_i$ are scheduled on machines $1$ through $i-1$.%
\label{lem:nodup-opt}
\end{lemma}
\begin{proof}
    Suppose some set $W_v \in V_i$ is placed on a machine $j > i$. By our ordering, $s_j m_j < s_i m_i$, so $W_v$ takes time at least $\max\{|W_v|/s_j m_j, ($length of longest chain in $W_v)/s_i\} > T$. This entails that the schedule has length greater than $T$.
\end{proof}

The algorithm then proceeds as follows. For each machine $i$, iterate through sets $V_{i-1}$ through $V_m$. For each set $V_{j}$, keep placing sets $W_v \in V_j$ on $i$ until either the number of jobs on $i$ exceeds $T s_i m_i$ or $V_j$ is empty. If, after all iterating through all machines, there are any sets $\tilde{W}_v$ not placed, then indicate that $\opt > T$.

\begin{lemma}
    If $\opt \le T$ then the algorithm places all sets $W_v$. 
\label{lem:nodup-Vi}
\end{lemma}
\begin{proof}
    Suppose, for some $v,i$ there is some $W_v \in V_i$ not placed by the algorithm. This can happen only if all machines $j =1$ through $i-1$ have load greater than $T s_j m_j$. By Lemma~\ref{lem:nodup-opt}, $\opt > T$. 
\end{proof}

\begin{lemma}
    The maximum number of jobs placed on any machine $i$ is at most $2T s_i m_i$. 
\label{lem:nodup-load}
\end{lemma}
\begin{proof}
    Since each $\tilde{W}_v$ placed on $i$ has size less than $T s_i m_i$, and the algorithm stops placing jobs placing jobs on $i$ as soon as its number of jobs exceeds $T s_i m_i$, the total number of jobs on $i$ does not exceed $2 T s_i m_i$.
\end{proof}

Let $U_i$ be the set of all jobs placed on machine $i \in S$. Let $U$ be the set of remaining jobs (not yet placed on either $S$ or $S'$). Partition $U_i$ into $U_{i,1}$ and $U_{i,2}$ where $U_{i,1} = \{u \in U_i : \exists v \in U_i \cap \Vlong,\ u = v$ or $u \prec v\}$. The algorithm then proceeds as follows.  First, for each $i$, schedule $U_{i,1}$ on $i$ in time $1$ to $t_1$. Second, use the algorithm from Section~\ref{sec:mdps-no-dup} to schedule $U$ on $S$ in time $t_1$ to $t_2$. Finally, for each $i$, schedule $U_{i,2}$ on $i$ in time $t_2$ to $t_3$. 

\begin{lemma}
    If $\opt \le T$, the constructed schedule has length $\polylog(n) \cdot T$. 
\end{lemma}
\begin{proof}
    By construction of $U_{i}$, $U_{i,1}$ is downward closed (i.e.\ if $v \in U_{i,1}$ and $u \prec v$, then $u \in U_{i,1}$) and $U_{i,2}$ is upward closed (i.e.\ if $v \in U_{i,2}$ and $v \prec u$ then $U \in U_{i,2}$). By construction of $V_i$, Lemma~\ref{lem:nodup-load}, and Graham's list scheduling theorem \cite{graham:schedule}, we can schedule $U_{i_1}$ and $U_{i,2}$ on $i$ in time $3T$. The lemma then follows from Lemma~\ref{lem:Lenstra} and Lemma~\ref{lem:no-dup}.
\end{proof}

This completes the proof of Lemma~\ref{lem:no-dup-symmetric} and the second part of Theorem~\ref{thm:no-dup}.


\section{\umps\ Reduces to No-Duplication Scheduling under Job-Machine Delays}
\label{sec:umps_reduction}
\paragraph{Job-Machine Delay Scheduling.}
We are given a set of precedence-ordered jobs and a set of machines. For each job-machine pair $(v,i)$ there is an associated delay $\delay{v,i}$. In a schedule, if $v$ is executed on machine $i$ at time $t$, then for any $u \prec v$, either $u$ is executed on $i$ before time $t$ or $u$ is executed on some machine $j$ before time $t - \delay{v,i}$. The objective is to construct a no-duplication schedule that minimizes makespan. 

\paragraph{\umps\ problem statement.}
We are given a set $V$ of $n$ unit-size, precedence ordered jobs, and a set $M$ of $m$ identical machines $M$. $V$ is partitioned into $m$ nonoverlapping sets, $V_i$ for each $i \in M$. In a valid schedule, if job $u$ precedes job $v$ then $u$ must be executed before $v$. Also, for each job $v$ and each machine $i$, if $v \in V_i$ then $v$ must be executed on $i$. The objective is to construct a no-duplication schedule that minimizes makespan.

\begin{theorem}
    If there is a polynomial time algorithm that computes an $\alpha$-approximation for job-machine delays, then there is a polynomial time algorithm that computes a $6\alpha$-approximation for \umps.
\label{thm:umps}
\end{theorem}
\begin{proof}
    Suppose we are given an instance $I$ of \umps. We construct an instance $I'$ of job-machine delays as follows. The set of machines is the same as $I$. For each machine $i$ we introduce two new jobs $u_i$ and $v_i$. All other jobs are identical to $I$. For all $v \in V_i$, we introduce precedence constraints such that $u_i \prec v \prec v_i$. Other precedence relations are the same as $I$. For all $i$ and $v \in V_i \cup \{u_i, v_i\}$, we set $\delay{v,i} = 0$ and $\delay{v,j} = n$ if $j \ne i$. 
    
    Let $\opt(I)$ and $\opt(I')$ be the optimal makespans of the $I$ and $I'$, respectively. Let $\sigma'$ be a job-machine delays schedule with makespan $C$.  If $C \ge n$, we can trivially construct a length $n$ \umps\ schedule of $I$ by list scheduling any available job on its assigned machine at each step. So we suppose $C < n$.  

    \begin{thmclaim}
        For each $i$, $V_i$ can be partitioned into two sets $V_{i,1}$ and $V_{i,2}$ with the following properties:
        \begin{enumerate*}[label=(\arabic*)]
            \item there is a machine $\mu_i$ such that $\sigma'$ executes every job in $V_{i,1}$ on $\mu_i$ with all its predecessors, and 
            \item $\sigma'$ executes every job in $V_{i,2}$ on $i$.
        \end{enumerate*} 
    \label{clm:Vi-partition}
    \end{thmclaim}
    \begin{proof}
        Fix machine $i$. We prove the claim by cases depending on where $u_i$ and $v_i$ are executed. 
        
        Suppose $u_i$ and $v_i$ are both executed on $i$. If any job in $V_i$ is executed on a machine other than $i$, then the schedule requires a communication of length $n$ which contradicts our assumption. So all jobs in $V_i$ are executed on $i$. In this case, $V_{i,1} = \varnothing$, $V_{i,2} = V_i$, and $\mu_i = i$. 

        The case where $u_i$ is executed on $i$ and $v_i$ is executed on $j \ne i$ is not possible because it requires a communication length $n$, which contradicts our assumption.

        Suppose $u_i$ is executed on $j \ne i$ and $v_i$ is executed on $i$. If any job in $V_i$ is executed on any machine $i' \not\in\{i, j\}$, then the schedule must have a communication of length $n$, which contradicts our assumption. So all jobs are executed on either machine $i$ or $j$. Consider a job $u$ that precedes a job $v \in V_i$ where $v$ is executed on $j$. If $u$ is not executed on $j$, then there must be a length $n$ communication, which contradicts our assumption. Letting $V_{i,1}$ be the set of jobs in $V_i$ executed on $i$, $V_{i,2}$ the set of jobs in $V_i$ executed on $j$, and $\mu_i = j$, we have that all jobs preceding a job in $V_{i,2}$ are also executed on $\mu_i$. 

        Suppose $u_i$ and $v_i$ are both executed on some machine $j \ne i$. If any job in $V_i$ is executed on any machine other than $j$, the schedule must have a communication of length $n$, contradicting to our assumption. Similarly for any jobs that precede jobs in $V_i$. So we may conclude that each job in $V_i$ along with its predacessors are all executed on $j$.  In this case, we set $\mu_i = j$, $V_{i,1} = V$, and $V_{i,2} = \varnothing$.  
    \end{proof}

    We now construct a schedule $\sigma$ of problem instance $I$ of \umps. By Claim~\ref{clm:Vi-partition}, we partition each $V_i$ into two sets, $V_{i,1}$ and $V_{i,2}$. For all jobs $v \in V_{i,1}$, if $\sigma'(v) = \{(j,t)\}$ then we set $\sigma(v) = \{(i,t)\}$.  Next, for each $i$ and all jobs $v \in V_{i,2}$, if $\sigma'(v) = \{(i,t)\}$ then we set $\sigma(v) = \{(i, C + t + 1)\}$. Informally, for each $i$, $\sigma$ schedules all jobs in $V_{i,1}$ on $i$ in time 1 to $C$, and all jobs in $V_{i,2}$ on $i$ in time $C + 1$ to $2C$. By Claim~\ref{clm:Vi-partition}, this is a valid schedule of $I$. 

    \begin{thmclaim}
        $\opt(I') \le \opt(I) + 2$.
    \end{thmclaim}
    \begin{proof}
        Suppose there is a schedule $\sigma_I$ of $I$ with makespan $D$. We define a schedule $\sigma_{I'}$ of $I'$ with makespan $D+2$. Let $\sigma_{I'}$ execute each $u_i$ on machine $i$ at time 1, then run schedule $\sigma_I$ from time 2 to $D+1$, and then execute each $v_i$ on $i$ at time $D+2$.  The construction entails that, for all $i$, $\sigma_{I'}$ schedules each job $v \in V_i \cup \{u_i, v_i\}$ on machine $i$. By definition of the delays in $I'$, there are no communication delays incurred in $\sigma_{I'}$. Therefore, $\sigma_{I'}$ is a valid schedule of $I'$. 
    \end{proof}

    Therefore, the makespan of $\sigma$ is $2\alpha\cdot\opt(I') \le 2\alpha\cdot (\opt(I) + 2) \le 6 \alpha \cdot \opt(I)$.
\end{proof}

\section{Reduction from Out/In-Delays to In-Delays}
\label{sec:out-in}
In this section, we show that we can achieve an approximation for machine delays and job precedence delays if we are given an algorithm to solve the problem when all out-delays are 0. Let $\I$ be the set of problems with machine delays and job precedence delays. Let $\I'$ be the subset of $\I$ for which all machine out-delays, and all job out-delays, are 0. We show that, given an algorithm to approximate any instance in $\I'$, we can approximate any instance in $\I$. For a given instance $I$, we define $\opt(I)$ to be the optimal makespan of $I$.  

\paragraph{Model.}
Most general model: machines have speed, size, and in/out-delays. Jobs have in/out-delays. For arbitrary $a,b$, we define $\groupdelay{a,b} = 2^a + 2^b$ and define an ordering on pairs: $(a,b) < (a',b')$ if $\groupdelay{a,b} \le \groupdelay{a',b'}$, breaking ties arbitrarily. 

\begin{lemma}
    There exists an algorithm which, on arbitrary input $I \in \I$ with delays given by $\delay{}$, produces a schedule with makespan at most $\alpha(\log \delay{\max})^4 \opt(I) + 16 \alpha \beta \delay{\max} (\log \delay{\max})^3$ if there exists an algorithm which, on arbitrary input $I' \in \I'$, produces a schedule with makespan at most $\alpha\cdot\opt(I') + \beta$.
\label{lem:reduction}
\end{lemma}
\begin{proof}
    We suppose we are given an instance $I \in \I$ with delays given by $\delay{}$. We construct an instance $I' \in \I'$ with delays given by $\ddelay{}$. $I'$ is identical to $I$, except that, for all $v$, $\ddelayout{v} = 0$ and $\ddelayin{v} = \delayin v + \delayout v$. Since out-delays are 0 in $I'$, we refer to $\ddelayin v$ and $\ddelayout i$ as $\ddelay v$ and $\ddelay i$, respectively. We then apply the supposed algorithm to solve $I'$ to produce a schedule $\sigma'$ with makespan at most $\alpha \opt(I') + \beta$.   

    To convert $\sigma'$ into a schedule $\sigma$ of $I$, we partition jobs into groups based on their delay: job $v$ is in group $V_{k,k'}$ if $2^{k-1} \le \delayin v < 2^k$ and $2^{k'-1} \le \delayout v < 2^{k'}$. Letting $K = \log \max_v\{ \delayin v, \delayout v\}$, we have at most $K^2$ job groups. We also define the group $U_{k,k'} = \{u : \exists v \in V_{k,k'},\ u \prec v\}$. Let $\sigma$ be the output of Algorithm~\ref{alg:reduction} on input $\sigma'$. (It is easy to see that Algorithm~\ref{alg:reduction} runs in polynomial time.) 
    
    \begin{alg}[Reduction: In-Delays to In/Out-Delays]
        Input: a schedule $\sigma'$. Initialize: $\tilde{\sigma}(v) = \varnothing$ and $\sigma(v) = \varnothing$ for all $v$. Output: schedule $\sigma$. 
        
        For all $v,i,t$, $\tilde{\sigma}(v) = \{(i,t + \delay{\max} - \delayout i): (i,t) \in \sigma'(v)\}$. Let 
        \begin{align}
            &S_{\ell,\ell',d} = \{v:  \delayin v \in [2^{\ell-1}, 2^{\ell}) \land \delayout v \in [2^{\ell'-1}, 2^{\ell'}) \land \exists i,t \in [d\groupdelay{\ell,\ell'}, (d+1)\groupdelay{\ell,\ell'}),\ (i,t) \in \tilde{\sigma}(v) \} \label{eq:Sll'd}\\
            &R_{\ell,\ell',d} = \{u: \exists v \in S_{\ell,\ell',d},\ u \prec v \land \exists i,t \in [d\groupdelay{\ell,\ell'}, (d+1)\groupdelay{\ell,\ell'}),\ (i,t) \in \tilde{\sigma}(u) \} \label{eq:Rll'd}
        \end{align}
        For all $v,i,\ell,\ell',t,d$, if $(i,t) \in \tilde{\sigma}(v)$ and $d\groupdelay{\ell,\ell'} \le t < (d+1)\groupdelay{\ell,\ell'}$ and $v \in S_{\ell,\ell',d} \cup R_{\ell,\ell',d}$, then we define
        \begin{equation}
            \tau = \sum_{(d,d')} \ceil{ (t-1)/ \groupdelay{d,d'}} (\groupdelay{d,d'} + 2^{d'}) + \sum_{(d,d') > (\ell,\ell')} \ceil{t/\groupdelay{d,d'}} (\groupdelay{d,d'} + 2^{d'})
            \label{eq:red_phase_expansion}
        \end{equation}
        and assign $\sigma(v) \gets \sigma(v) \cup \{(i,\tau + t - d \groupdelay{\ell,\ell'} )\}$. 
    \label{alg:reduction}
    \end{alg}
    
    We provide an informal description of Algorithm~\ref{alg:reduction}. 
    Let phase $\phi_{\ell,\ell',d} = S_{\ell,\ell',d} \cup R_{\ell,\ell',d}$ and let the start time of $\phi_{\ell,\ell',d}$ be $d\groupdelay{\ell,\ell'}$. The schedule is structured so that each phase $\phi_{\ell,\ell'}$ of $\tilde{\sigma}$ is allotted $2\groupdelay{\ell,\ell'}$ time in $\sigma$, during which no other phases are executed. If phase $\phi_{\ell,\ell'}$ is alloted time $t$ to $t + 2 \groupdelay{\ell,\ell'}$ in $\sigma$, the jobs are executed from time $t + 2^{\ell}$ to time $t + 2^{\ell} + \groupdelay{\ell,\ell'}$. (Jobs are always executed in the same order, and on the same machines, as in $\tilde{\sigma}$.) This allows for an initial in-communication phase of length $2^{\ell}$ and a final out-communication phase of length $2^{\ell'}$. Earlier phases of $\tilde{\sigma}$ are exected in $\sigma$ before later phases of $\tilde{\sigma}$. Also, if multiple phases have the same start time in $\tilde{\sigma}$ then we execute them in non-increasing order of the sum of their delays: if $\phi_{\ell_1, \ell_2}$ and $\phi_{\ell_3,\ell_4}$ have the same start time in $\tilde{\sigma}$ and $\phi_{\ell_1, \ell_2}$ is executed before $\phi_{\ell_3,\ell_4}$ in $\sigma$ then $\ell_1 + \ell_2 \le \ell_3 + \ell_4$. Figure~\ref{fig:reduction}(a) depicts the structure of $\sigma$ for two consecutive values of $(\ell,\ell')$, where the corresponding phases have the same start time.

    \junk{
    \begin{algorithm}
        \SetKwInput{Input}{input}
        \SetKwInput{Init}{init}
        \Input{schedule $\sigma$}
        \Init{$\forall v,\ S_v \gets \varnothing$ and $\sigma'(v) \gets \varnothing$;\ $T \gets 0$; $T' \gets 0$}
        $\forall v,i,t$,\ if $(i,t) \in \sigma(v)$ then $S_v \gets S_v \cup \{(i,t + \delay{\max} - \delayout i)\}$ \label{line:reduce-machine}\;
        \While{$T \le C$}{
            \For{($k=K$ to $0$) and ($k'=K$ to $0$): $\exists$ integer $a,\ a(2^k + 2^{k'}) = T $}{
                \ForAll{$i,t,v : v \in (V_{k,k'} \cup U_{k,k'})$}{
                    \If{$(i,t) \in S_v$ and $T \le t < 2^{k+1}$}{
                        $\sigma'(v) \gets \sigma'(v) \cup \{(i, T' + 2^k + t - T)\}$ \label{line:reduce-job}\;
                    }
                }
                $T' \gets T' + 2(2^k + 2^{k'})$\;
            }
            $T \gets T+1$
        }
    \caption{Reduction}
    \end{algorithm}
    }

    \begin{claim}
        $\sigma$ is a valid schedule of $I$. 
    \end{claim}
    \begin{proof}
        Since each job is executed on the same machines as in $\tilde{\sigma}$, and the relative ordering of jobs on each machine is maintained, we can immediately infer that precedence relations are satisfied on a single machine. So we must only verify that communication constraints are met.
        
        Let $u,v$ be a precedence ordered pair with $u \prec v$. Suppose some instance of $v$ is executed on $i$ at time $t_v$ in $\sigma'$, and $u$ is executed on $j$ at time $t_u < t_v - \delayin v - \delayout v - \delayin i - \delayout i$. Let $\phi_{a,b,c} \ni v$ (resp., $\phi_{e,f,d} \ni u$) be the first phase in which $v$ (resp., $u$) is executed.  Note that $\groupdelay{a,b} \ge \delayin v + \delayout v$ and $\groupdelay{e,f} \ge \delayin u + \delayout u$.  Let $t^{\mathrm{new}}_v$ and $t^{\mathrm{new}}_u$ be the time at which these instances of $u,v$ are executed in $\sigma'$. Let $\tau_v$ and $\tau_u$ be the values of $\tau$ given in equation~(\ref{eq:red_phase_expansion}) for $v,i,a,b,t_v,c$ and $u,j,e,f,t_u,d$, respectively. 
        \begin{align*}
            t^{\mathrm{new}}_v &\ge t_v + \delay{\max} - \delayout i + \tau_v - c \groupdelay{a,b} &&\by{construction of } \sigma \\
            &\ge t_u + \delayin v + \delayout v + \delayin i + \delayout i + \delay{\max} - \delayout i + \tau_v - c \groupdelay{a,b} &&\by{supposition} \\
            &\ge t^{\mathrm{new}}_u + \delayout j - \tau_u + d \groupdelay{e,f} + \delayin v + \delayout v + \delayin i + \tau_v -  c \groupdelay{a,b} &&\by{construction of } \sigma \\
            &\ge t^{\mathrm{new}}_u + \delayout j + (2^e + 2^f) + d \groupdelay{e,f} + \delayin v + \delayout v + \delayin i - c \groupdelay{a,b} &&\by{} \tau_v - \tau_u \ge 2^e + 2^f\\
            &\ge t^{\mathrm{new}}_u + \delayout j + \delayout u + \delayin v + \delayin i &&\by{} 2^e + 2^f \ge c\groupdelay{a,b} - d\groupdelay{e,f} 
        \end{align*}
        This proves the claim.
    \end{proof}

    \begin{figure}
        \centering
        \includegraphics[width=\textwidth]{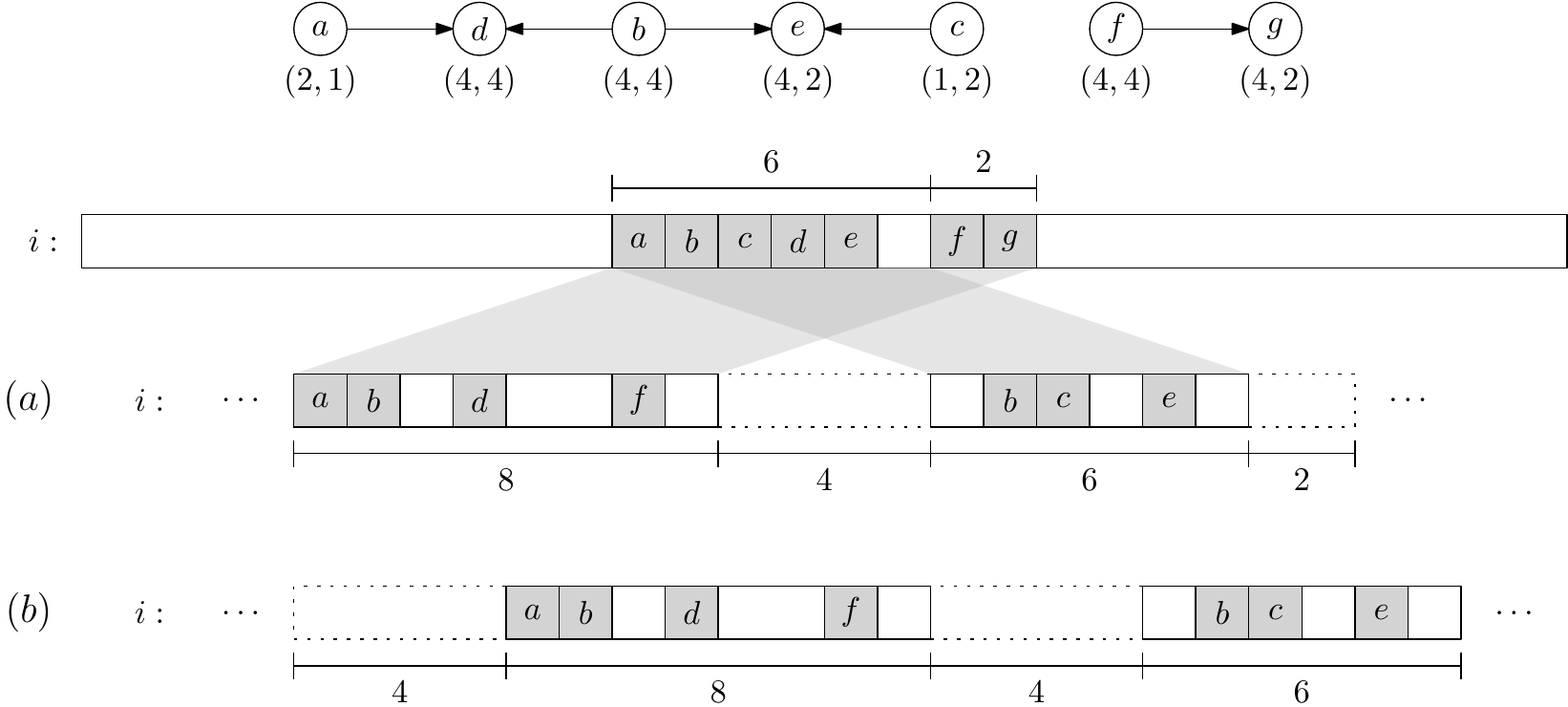}
        \caption{The expansion of a schedule via Algorithms~\ref{alg:reduction} and \ref{alg:reduction2}. (a) Algorithm~\ref{alg:reduction} adds an out-communication period to the end of each phase. (b) Algorithm~\ref{alg:reduction2} adds an in-communication period to the start of each phase.}
        \label{fig:reduction}
    \end{figure}
    
    \begin{claim} 
        The makespan of $\sigma$ is at most $(\log \delay{\max})^2(\alpha \opt(I') + \beta) + 8 \delay{\max}\log \delay{\max}$.
    \label{clm:red_makespan}
    \end{claim}
    \begin{proof}
        By construction, the length of $\sigma$ is at most $\sum_{(d,d')} \ceil{\frac{\alpha\opt(I') + \beta}{\groupdelay{d,d'}}} 2 \groupdelay{d,d'} \le (\log \delay{\max})^2(\alpha \opt + \beta) + 8 \delay{\max} \log \delay{\max}$, using $\groupdelay{d,d'} = 2^d + 2^{d'}$. 
    \end{proof}

    \begin{claim}
        $(\log \delay{\max})^2 \opt(I) + 8 \delay{\max}\log \delay{\max} \ge \opt(I')$.
    \end{claim}
    \begin{proof}
        We define the following algorithm.
        \begin{alg}[Reduction: In/Out-Delays to In-Delays]
            Input: a schedule $\sigma$. Initialize: $\tilde{\sigma}(v) = \varnothing$ and $\sigma(v) = \varnothing$ for all $v$. Output: schedule $\sigma'$. 
            
            For all $v,i,t$, $\tilde{\sigma}(v) = \{(i,t + \delay{\max} + \delayout i): (i,d) \in \sigma(v)\}$. Let $S_{\ell,\ell',d}$ and $R_{\ell,\ell',d}$ be defined as in equations~\ref{eq:Sll'd} and \ref{eq:Rll'd}. For all $v,i,\ell,\ell',t,d$, if $(i,t) \in \tilde{\sigma}(v)$ and $d\groupdelay{\ell,\ell'} \le t < (d+1)\groupdelay{\ell,\ell'}$ and $v \in S_{\ell,\ell',d} \cup R_{\ell,\ell',d}$, then we define
            \begin{equation}
                \tau = \sum_{(d,d')} \ceil{ (t-1)/ \groupdelay{d,d'}} (\groupdelay{d,d'} + 2^{d}) + \sum_{(d,d') > (\ell,\ell')} \ceil{t/\groupdelay{d,d'}} (\groupdelay{d,d'} + 2^{d})
                \label{eq:red_phase_expansion2}
            \end{equation}
            and assign $\sigma'(v) \gets \sigma'(v) \cup \{(i,\tau + t - d \groupdelay{\ell,\ell'} + 2^d)\}$. 
        \label{alg:reduction2}
        \end{alg}
        
        Informally, this Algorithm~\ref{alg:reduction2} is similar to Algorithm~\ref{alg:reduction}. One difference that, when constructing $\tilde{\sigma}$, this algorithm \textit{increases} the execution times on each machine $i$ by $\delayout i$. The other main difference is that when executing each phase $\phi_{d,d'}$, this algorithm adds an in-communication period before the execution of the phase. Figure~\ref{fig:reduction}(b) depicts the structure of $\sigma'$ for two consecutinve values of $(\ell,\ell')$ when the corresponding phases have the same start time.   
            
        Suppose we are given an arbitrary schedule $\sigma_1$ of $I$ with makespan $C$. Let $\sigma_2$ be the result of running Algorithm~\ref{alg:reduction2} on input $\sigma_1$.  By the same reasoning in claim~\ref{clm:red_makespan}, we can see that the length of $\sigma_2$ is a most $(\log \delay{\max})^2 \opt(I) + 8 \delay{\max}\log \delay{\max}$. This proves the claim
    \end{proof}
    The claims are sufficient to prove the lemma.
\end{proof}

\begin{corollary}
    When all job delays are 0 there exists an algorithm which, on arbitrary input $I \in \I$ with delays given by $\delay{}$, produces a schedule with makespan at most $\alpha \opt(I) + \beta + (\alpha + 1)\delay{\max}$ if there exists an algorithm which, on arbitrary input $I' \in \I'$, produces a schedule with makespan at most $\alpha\cdot\opt(I') + \beta$.
\label{cor:reduction_machine}
\end{corollary}
\begin{proof}
    The corollary follows from the proof of Lemma~\ref{lem:reduction} when job delays are 0. Specifically, using Algorithms~\ref{alg:reduction} and \ref{alg:reduction2}, we can infer that the makespan of $\sigma$ is at most $\alpha\opt(I') + \beta + \delay{\max}$ and that $\opt(I') \le \opt(I) + \delay{\max}$. This proves the corollary.
\end{proof}

\section{Algorithm for Uniform Machines}
\label{sec:uniform}

In this section, we show the existence of an algorithm which achieves provably good bounds for scheduling on uniform machines with fixed communication delay. The algorithm presented here is a generalization of the algorithm in \cite{LR02}, the main difference being our incorporation of parallel processors for each machine. 

Recall that \udps\ is identical to \mdps\ with $\delayout{i} = 0$ for all $i$, and with $\rho_i = \rho_j$, $m_i = m_j$, and $s_i = s_j$ for all $i,j$. We define \udps-Solver as \Cref{alg:uniform}. The algorithm takes as input a set of jobs $U$ and a group $\group k$ of identical machines with delay $\delay{}$, speed $\groupspeed k$, and size $\groupsize k$. We are guaranteed that the length of the longest chain in $U$ is at most $\alpha \delay{} \groupspeed k$ and that, for any job $v \in U$, the number of jobs $u \in U$ such that $u \prec v$ is at most $\alpha \delay{} \groupsize k \groupspeed k$.

\begin{algorithm}
\KwData{set of jobs $U$, set of identical machines $\group k$, delay $\delay{}$}
\KwInit{$t \gets 0$\;}
\While{$U \ne \varnothing$}{
    $\forall i \in \group k,\  V_{i,t} \gets \varnothing$ \label{line:Vit_init}\;
    \ForAll{jobs $v \in U$}{
        $U_{v,t} \gets \{u \in U: u \prec v\}$\;
        $U^{\dup}_{v,t} \gets U_{v,t} \cap \{u \in U : u \prec v$ and $\exists i, u \in V_{i,t} \}$\;
        \If{$|U_{v,t}| \ge 2 \cdot |U^{\dup}_{v,t}|$ \label{line:overlap}}{
            $i \gets \arg\min_{j \in \group k} \{ | V_{j,t} | \}$ \label{line:machine}\;
                $V_{i,t} \gets V_{i,t} \cup U_{v,t}$ \label{line:place}\;
        }
    }
    $\forall i \in \group k$, list schedule $V_{i,t}$ on $i$ starting at time $t$\;
    $U \gets U \setminus \big( \bigcup_i V_{i,t} \big)$ \label{line:remove}\;
    $t \gets \delay{} + \max\{ t, \max_{v} \{$completion time of $v\}\}$ \label{line:update_t}\;
}
\caption{\udps-Solver}
\label{alg:uniform}
\end{algorithm}

\begin{lemma}
    \Cref{alg:uniform} produces a valid schedule of $U$ on $\group k$ in polynomial time.
\end{lemma} 
\begin{proof}
    For fixed $i,t$, the list scheduling subroutine guarantees that precedence constraints are maintained in placing $V_{i,t}$. So we must only show that precedence and communication constraints are obeyed across different values of $t$. For a fixed value $t^*$, line~\ref{line:remove} entails that all jobs in $\bigcup_i V_{i,t}$ for $t < t^*$ have been removed from $U$. This entails that all precedence constraints are obeyed in scheduling $V_{i,t^*}$. The update at line~\ref{line:update_t} ensure that no jobs are executed for $\delay{}$ before scheduling $V_{i,t^*}$. Since there is no communication necessary for scheduling $\bigcup_i V_{i,t^*}$, this entails that communication constraints are obeyed.  
\end{proof}

\begin{lemma}
    For any $t$, $\sum_i |V_{i,t}| \le 2 \cdot |\bigcup_i V_{i,t}|$.
    \label{lem:dup_load}
\end{lemma}
\begin{proof}
    Consider a fixed value of $t$. For any machine $i$, let $V_{i,t,0}$ be the value of $V_{i,t}$ after initialization at line~\ref{line:Vit_init}. Also, for all machines $i$, let $V_{i,t,\ell}$ be the value of $V_{i,t}$ after the $\ell$\textsuperscript{th} time some set $V_{j,t}$ is updated.  We prove the lemma by induction on $\ell$. 
    
    The lemma holds trivially for $\ell = 0$ since $\sum_i |V_{i,t,0}| = |\bigcup_i V_{i,t,0}| = 0$. Suppose the lemma holds up until the $\ell$\textsuperscript{th} update of any $V_{i,t}$, and consider the $(\ell + 1)$\textsuperscript{th} update. Let $U'$ be the set of jobs placed during the update. The condition at line~\ref{line:overlap} entails that at least half the jobs in $U$ are not in $\bigcup_i V_{i,t,\ell}$. Therefore,
    \begin{align*}
        |\bigcup_i V_{i,t,\ell+1}| &= |U \cup \bigcup_i V_{i,t,\ell}| = |U| + |\bigcup_i V_{i,t,\ell}| - |U \cap \bigcup_i V_{i,t,\ell}| \\
        &\ge |U| + |\bigcup_i V_{i,t,\ell}| - \frac{|U|}{2} &&\by{supposition}\\
        &\ge \frac{|U|}{2} + \frac{\sum_i |V_{i,t,\ell}|}{2} \ge \frac{1}{2} \sum_i |V_{i,t,\ell+1}| &&\by{IH} 
    \end{align*}
    This proves the lemma.
\end{proof}

\begin{lemma}
    If $t_1$ and $t_2$ are two consecutive values of $T$, then $t_2 \le t_1 + \frac{2 \cdot |\bigcup_i V_{i,t}|}{|\group k| \cdot \groupsize k \groupspeed k} + 3\alpha \delay{}$. 
\label{lem:load_balance}
\end{lemma}
\begin{proof}
    Consider a fixed value of $t$. If there are no jobs added to $V_{i,t}$ for any $i$, then $t_2 = t_1 + \delay{}$ by line~\ref{line:update_t} and the lemma holds. So suppose that $\max_v\{$completion time of $v\} > t_1$ on execution of line~\ref{line:update_t}. At the execution of this line, suppose the last job completes on machine $i^*$. Then 
    \begin{align*}
        t_2 &= \delay{} + t_1 + \max_v\{\text{completion time of } v \text{ on } i^*\} - t_1 &&\by{line~\ref{line:update_t}}\\
        &\le \delay{} + t_1 + \frac{|V_{i^*,t}|}{\groupsize k \groupspeed k} + (\text{length of longest chain in } U) &&\by{Graham \cite{graham:schedule}} \\
        &\le \delay{} + t_1 + \frac{|V_{i^*,t}|}{\groupspeed k \groupsize k} + \alpha \delay{} &&\by{supposition} \\
        &\le \delay{} + t_1 + \min_i \Big\{ \frac{|V_{i,t}|}{\groupsize k \groupspeed k} \Big\} + \frac{\alpha \delay{} \groupsize k \groupspeed k}{\groupsize k \groupspeed k} + \alpha \delay{} &&\by{supposition and line~\ref{line:machine}} \\
        &\le  t_1 + \frac{\sum_i |V_{i,t}|}{|\group k| \cdot \groupsize k \groupspeed k} + \delay{} (2 \alpha + 1) \le t_1 + \frac{2 \cdot |\bigcup_i V_{i,t}|}{|\group{k}| \cdot \groupsize k \groupspeed k} + \delay{} (2\alpha + 1) &&\by{Lemma~\ref{lem:dup_load}}
    \end{align*}
    The appeal to line~\ref{line:machine} invokes the fact that the algorithm always chooses the least loaded machine when placing new jobs. Since every job has a maximum of $\alpha \delay{} \groupsize k \groupspeed k$ predecessors in $U$, by supposition, the difference $|V_{i,t}| - |V_{j,t}|$ is always maintained to be less than  $\alpha \delay{} \groupsize k \groupspeed k$, for any $i,j$. This proves the lemma.
\end{proof}

\begin{lemma}
    Let $U$ be a set of jobs such that for any $v \in U$ the number of predecessors of $v$ in $U$ (i.e., $|\{ u \prec v\} \cap U|$) is at most $\alpha \delay{} \groupsize k \groupspeed k$, and the longest chain in $U$ has length at most $\alpha \delay{}$. Then given as input the set $U$ of jobs and the set $\group k$ of identical machines, \udps-Solver produces, in polynomial time, a valid schedule with makespan less than $3\log(\alpha \delay{} \groupsize k \groupspeed k)\alpha \delay{} + \frac{2 \cdot |U|}{|\group k| \groupsize k \groupspeed k} + \delay{}$.
\label{lem:udps}
\end{lemma}

\begin{proof}
    By the condition of line~\ref{line:overlap}, each time the variable $t$ increments we can infer that all unscheduled jobs have their number of unscheduled predecessors halved. By our supposition, this entails $t$ can be incremented only  $\log(\alpha \delay{} \groupsize k \groupspeed k)$ times before all jobs are scheduled. Let $t_{\ell}$ be $\ell$\textsuperscript{th} value of $t$, for $\ell = 0, 1, 2, \ldots, \log(\alpha \delay{} \groupsize k)$. 
    \begin{claim}
        For any $\ell$, we have  $\displaystyle t_{\ell + 1} \le 3\ell \alpha \delay{} +  \frac{2 \cdot |\bigcup_i \bigcup_{\ell' = 0}^{\ell} V_{i,t_{\ell'}}| }{|\group k| \cdot \groupsize k} + \bar{\rho_k}$
    \end{claim}
    \begin{proof}
        We prove the claim by induction on $\ell$. For $\ell = 0$, the claim holds since the inital value of $t$ is $0$. Suppose the claim holds up to $\ell$. Then 
        \begin{align*}
            t_{\ell + 1} &\le t_{\ell} + \frac{2 \cdot |\bigcup_i V_{i,t_{\ell}}|}{|\group k| \cdot \groupsize k} + 3 \alpha \delay{}  &&\by{Lemma~\ref{lem:load_balance}}\\
            &\le 3 \ell \alpha \delay{} + \frac{2 \cdot |\bigcup_i \bigcup_{\ell' = 1}^{\ell-1} V_{i, t_{\ell'}}|}{|\group k| \cdot \groupsize k} + \frac{2 \cdot |\bigcup_i V_{i,t_{\ell}}|}{|\group k| \cdot \groupsize k} + 3 \alpha \delay{}  &&\by{I.H.} \\
            &\le 3 \alpha \delay{}(\ell + 1) +  \frac{2 \cdot |\bigcup_i \bigcup_{\ell' = 1}^{\ell} V_{i, t_{\ell'}}|}{|\group k| \cdot \groupsize k} 
        \end{align*}
        where the last line follows from the fact that $\bigcup_i V_{i,t_{\ell}}$ and $\bigcup_i \bigcup_{\ell' = 1}^{\ell-1} V_{i, t_{\ell'}}$ have no members in common.
    \end{proof}
    The claim is sufficient to establish the lemma.
\end{proof}

\bibliographystyle{alpha}
\bibliography{refs.bib}

\end{document}